\documentclass[8pt,a4paper]{article}

\DeclareMathAlphabet{\mathtt}{OT1}{cmtt}{m}{n}
\DeclareFontShape{OT1}{cmtt}{m}{n}
     {%
      <5-8.99>cmtt8<9-9.99>cmtt9%
      <10-11.99>cmtt10%
      <12-25>cmtt12%
      }{}
\RequirePackage{amsfonts}

\newif\iftrimmarks  \trimmarksfalse \trimmarkstrue

\newdimen\papertrimheight
\newdimen\papertrimwidth
\newcommand{\anonymize}[2]{#1}
\long\def\anonymizelong #1\endanonymizelong{#1}

\usepackage{indentfirst}
\usepackage{epigraph}

\usepackage{overpic}
\usepackage{xcolor}
\usepackage{cmap}

\usepackage[english]{babel}
\usepackage{amsfonts,amssymb,amsmath,amscd,amsthm}
\usepackage{enumitem,hyperref}

\usepackage{cases}

\usepackage{pgf,tikz,circuitikz}\usetikzlibrary{arrows}
\usepackage{wrapfig}
\usepackage{geometry}
\usepackage{titling}

\usepackage{pb-diagram}

\usepackage{color}
\usepackage{xcolor}

\usepackage{pdfpages}

\usepackage{makecell}

\long\def\comment#1\endcomment{}
\long\def\solutions#1\endsolutions{#1}
\long\def\lktgonly#1\endlktgonly{}
\long\def\iumonly#1\endiumonly{}
\long\def\tmpcomment#1\endtmpcomment{}
\long\def\skipprint#1\endskipprint{#1}

\theoremstyle{theorem}
\newtheorem{theorem}{Theorem}
\newtheorem{lemma}{Lemma}

\newtheorem{example}{Example}
\newtheorem*{theorem2prime*}{Theorem~5$'$}

\theoremstyle{remark}

\theoremstyle{definition}
\newtheorem{definition}{Definition}
\newtheorem*{definition*}{Definition}

\newenvironment {th*}[1]
    {\gdef\thname{#1} \begin{thn}}%
    {\end{thn}}
\newtheorem*{thn}{\thname}

\begin{document}

\title{Feynman checkers: through the looking-glass}

\author{\anonymize{F. Ozhegov, M. Skopenkov, and A. Ustinov}{Anonymous authors}}

\date{}

\maketitle

\begin{abstract}
Feynman gave a famous elementary introduction to quantum theory by discussing the thin-film reflection of light. We make his discussion mathematically rigorous, keeping it elementary, using his other idea. The resulting model leads to accurate quantitative results and allows us to derive a well-known formula from optics. In the process, we get acquainted with mathematical tools such as Smirnov's fermionic observables, transfer matrices, and spectral radii. Quantum walks and the six-vertex model arise 
as the next step in this direction.


\end{abstract}






In his famous lectures~\cite{Feynman}, Feynman gave an elementary introduction to quantum theory by discussing the reflection of light by glass. The purpose of this paper is to make his discussion mathematically rigorous, keeping it elementary. This leads us to accurate quantitative results and allows us to derive a well-known formula from optics, given e.g. in another famous book~\cite{Landavshits-08}; 
 see Theorem~\ref{th-reflection} below.

Feynman's model is simple and we could start playing this ``game'' immediately.
Light entering the glass turns into a checker, just like Alice turns into a pawn in a known novel. But we prefer to start with a summary of one of Feynman's lectures (Sections~\ref{experiment}--\ref{theory}) to catch what it is all about. After that, we define a 
mathematical model (Section~\ref{definitions}) and prove the desired formula (Sections~\ref{Theorem}--\ref{proofs}). This brings us to other related models that are trendy now --- quantum walks \cite{Kempe-09, \anonymize{KSS-23, SU-22,}{} Venegas-Andraca-12} 
and the six-vertex model \cite{Duminil-Copin-etal-22} (Section~\ref{Problems}). Finally, with the most patient readers, we dive into technicalities postponed before (Appendix~\ref{appendix_correct}). Still, some natural questions will 
remain open or answered by references only.

\section{Experiment.} \label{experiment}

\epigraph{
Nearly everything is really interesting if
you go into it deeply enough.}{R.P. Feynman}

First, we talk about an experiment and its surprising results. This is a thought and idealized experiment, and we omit the important technical conditions required to carry it out in practice (see~\cite{Feynman}).

We need a few devices. The first one is a \emph{light source}. We need a rather special source, such as a laser, that emits light of just one particular color, say, red.

Our next device is a \emph{light detector}, or \emph{photomultiplier}. It produces a sound (click) when hit by light. For bright light, such clicks sum up to a constant noise. The photomultiplier is very sensitive, and it reacts even to weak light. If we decrease light intensity, then the clicks do not become quieter but rather happen less often. For very feeble light, we hear separate clicks with greater intervals between them. 

This is already a remarkable result: it demonstrates that \emph{light consists of particles}. Each particle of light called \emph{photon} carries some minimal amount of energy affecting the detector. For weaker light, there are fewer particles, but the energy of each one remains the same.

The last participant of our experiment is a piece of glass. We know that glass reflects light partially: looking at the window at night, we see both the objects outside and the reflection of a lamp inside. 
This means that light partially passes through the glass and is partially reflected. Roughly 4\% of incident light is reflected. 
In Figure~\ref{fig:exp1} to the left, one detector detects reflected light, and the other one (inside the glass) detects transmitted light. For weak enough light, they detect single photons. If photons are emitted one at a time, then the left detector clicks in 4\% cases and the right one clicks in 96\% cases; they never click simultaneously. As always, we neglect other effects such as small absorption of light. 
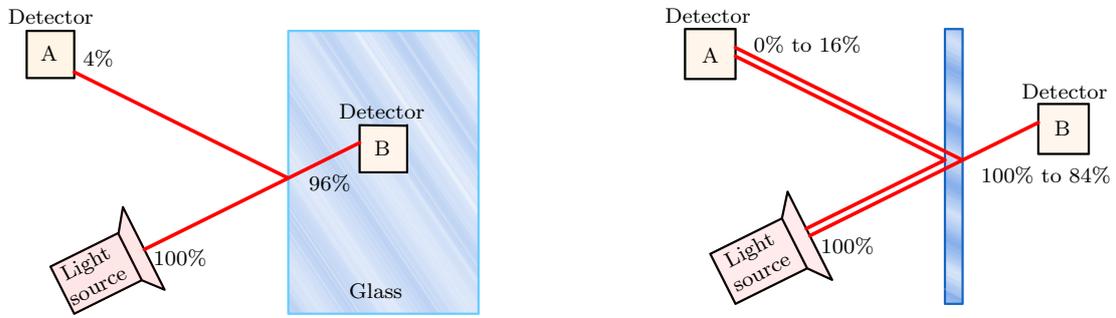
\begin{figure}[t]
\centering
\begin{tabular}{cc}
    &\begin{overpic}[width=0.4\linewidth]{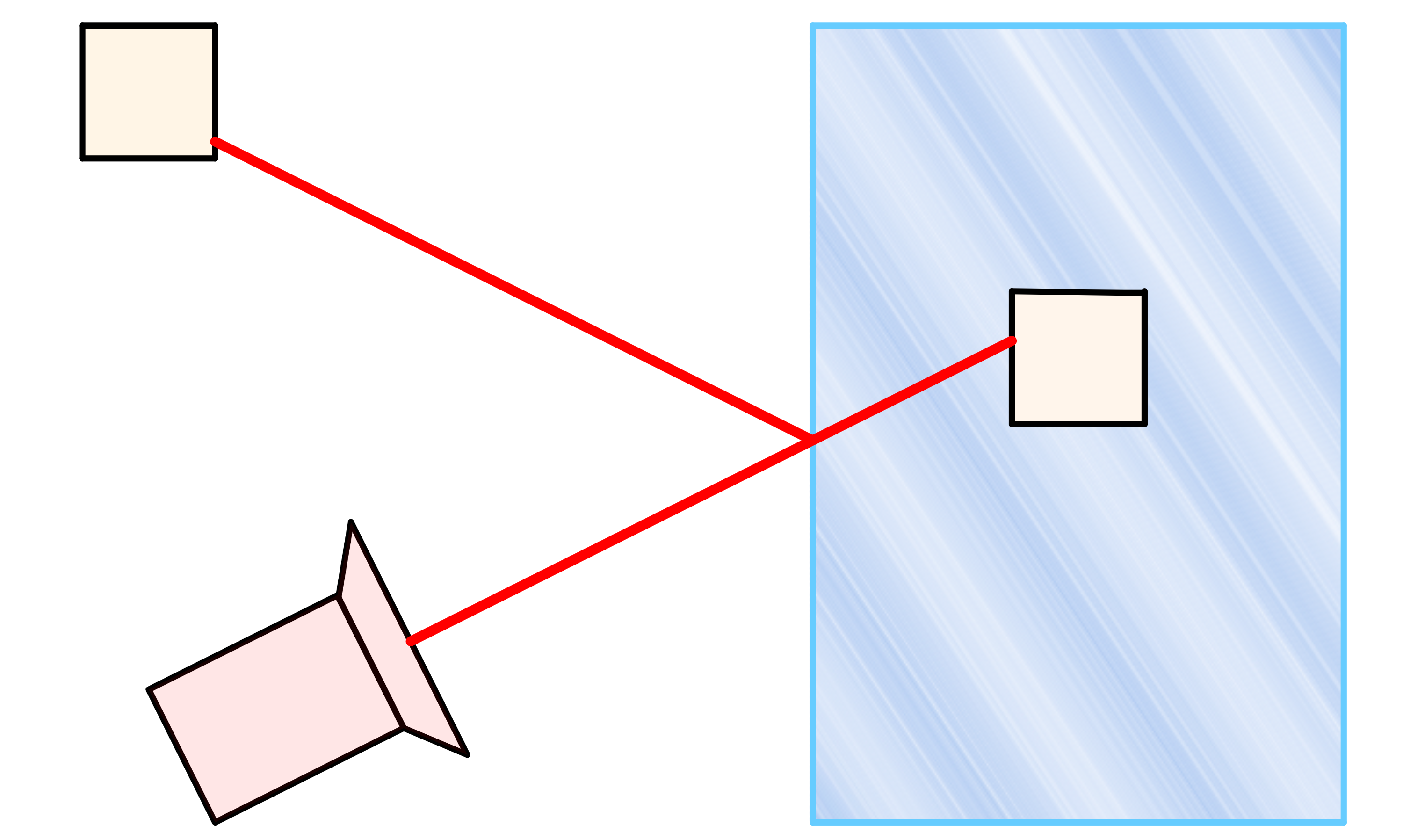}
    \footnotesize{
        \put(11, 5){\rotatebox{28}{\makecell{Light\\source}}}
        \put(8.5, 52){A}
       \put(2, 59.5){Detector}
       \put(75, 33){B}}
        \put(68, 40.5){Detector}
        \put(70, 4.5){Glass}
        \put(62, 26){$96\%$}
        \put(17, 51){$4\%$}
        \put(31, 11){$100\%$}
    \end{overpic}
    \hspace{2cm}
    \begin{overpic}[width=0.345\linewidth]{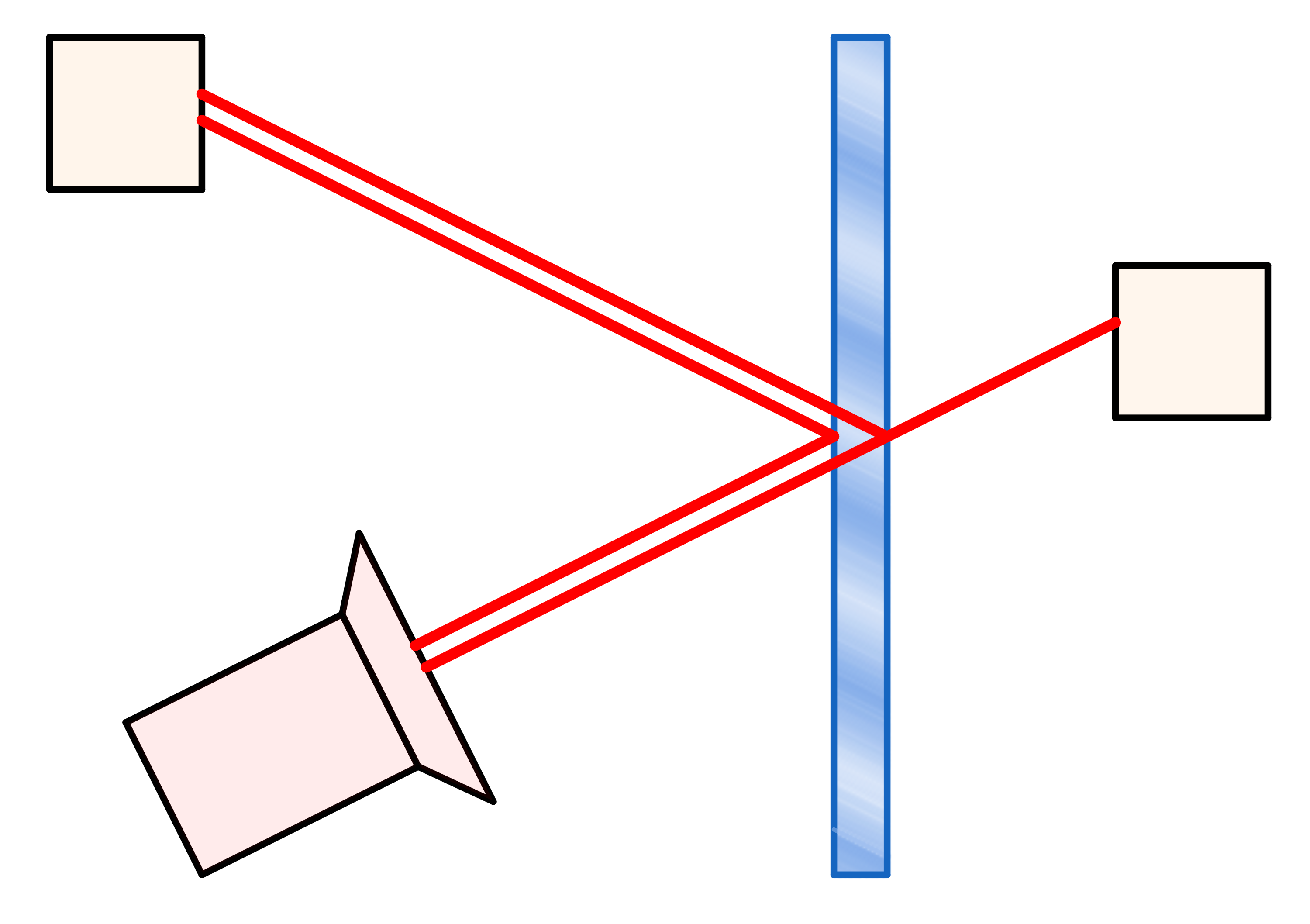}
    \footnotesize{
        \put(11, 9){\rotatebox{28}{\makecell{Light\\source}}}
        \put(7.5, 60){A}
        \put(-1, 69.25){Detector}
        \put(89, 43){B}
        \put(81.5, 51.5){Detector}
        \put(19.5, 62.5){$0\%$ to $16\%$}
        \put(72, 32){$100\%$ to $84\%$}
        \put(35, 15.5){$100\%$}
        }
    \end{overpic}
\end{tabular}
    \caption{(Left) An experiment to measure the partial reflection of light by a single surface of glass. About $4\%$ of light is reflected and hits Detector A while the rest $96\%$ is transmitted and hits Detector B. 
    (Right) 
    A similar experiment with two reflective surfaces. Depending on the thickness 
    of the glass, from $0\%$ to almost $16\%$ of light is reflected and hits Detector A. The additional surface can ``turn off'' or ``amplify'' reflection. This is a challenge for any reasonable theory.}
    \label{fig:exp1}
\end{figure}

This simple experiment is suggestive and rather hard to explain. How does the photon know if it should be reflected or transmitted? How does it happen that always 4\% of photons are reflected? 
The reader is cordially invited to pause here for a moment and come up with an explanation.

So, we proceed. One possible explanation could be that there is a hidden mechanism inside the photon, a kind of sheaves and gears, which determines if the photon is going to be reflected or not. In other words, there are two types of photons: 4\% are 
predisposed to be reflected, and the rest are 
predisposed to pass. 

Another possible explanation could be that there are small spots reflecting light on the surface of the glass, and the rest are holes allowing light to pass. At first sight, this looks reasonable: if 4\% of the surface are spots, then we get the 4\% reflection.

Now we discuss an experiment that disproves these and any other reasonable theory of partial light reflection.  
We just added a second reflective surface to our glass. That is, instead of a huge piece of glass (see Figure~\ref{fig:exp1} to the left), we take a thin glass with strictly parallel surfaces (see Figure~\ref{fig:exp1} to the right). The photon can now be reflected by the front surface or the back one. With already two theories of reflection at hand, let us 
make predictions.

The theory of reflective spots says that 4\% of the front surface are reflective spots. Assume that the same holds for the back surface. Then, after 4\% of light was reflected by the front surface, the rest reached the back surface, and 4\% of it was reflected (that is 4\% out of 96\% of transmitted light, or, roughly 4\%). In total, the theory 
predicts roughly 8\% of light to be reflected by two surfaces. 

The theory of hidden mechanism 
predicts 4\% of light to be reflected in total, because the photons that could reach the back surface, according to the theory, are all predisposed to pass. So, we have two competing theories that predict different values. 

\begin{figure}[b]
    \centering
    \begin{overpic}[width=0.9\linewidth]{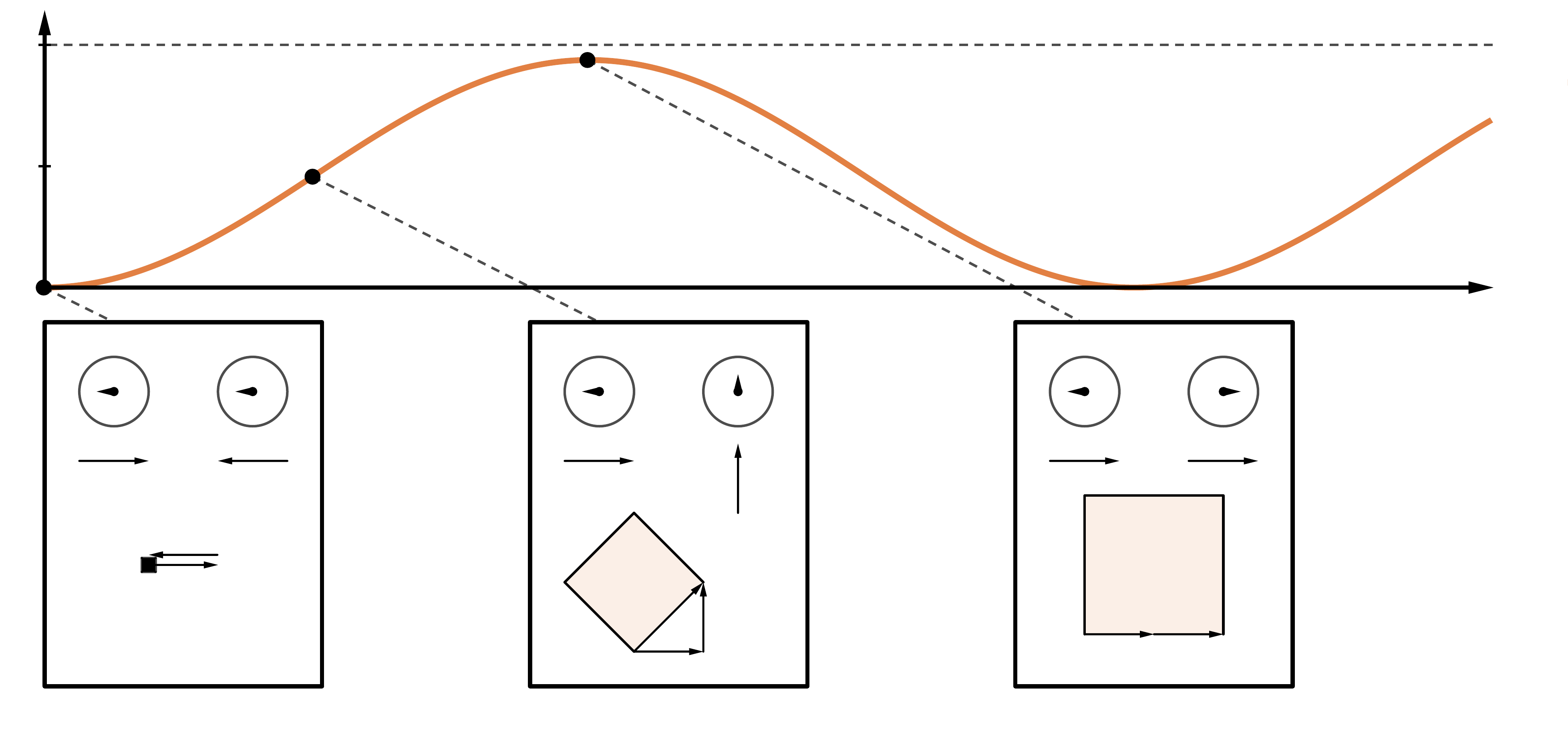}
    \footnotesize{
        \put(-2, 46.5){$\%$ of reflected light}
        \put(-0.5, 43){$16$}
        \put(0.5, 35.25){$8$}
        \put(94,25){Thickness} 
        \put(5.75, 17.75){$0.2$}
        \put(15, 17.75){$0.2$}
        \put(5, 24.25){front}
        \put(14, 24.25){back}
        \put(10, 11.75){$0.2$}
        \put(10, 8.5){$0.2$}
        \put(5.5, 10){$0\%$}
        \put(36.5, 17.75){$0.2$}
        \put(47.5, 15.25){$0.2$}
        \put(36, 24.25){front}
        \put(45, 24.25){back}
        \put(45, 6.25){$0.2$}
        \put(41.5, 3.5){$0.2$}
        \put(39.75, 8.5){$8\%$}
        \put(67.75, 17.75){$0.2$}
        \put(76.5, 17.75){$0.2$}
        \put(67, 24.25){front}
        \put(76, 24.25){back}
        \put(70, 4.5){$0.2$}
        \put(74.75, 4.5){$0.2$}
        \put(71.75, 10){$16\%$}
        }
    \end{overpic}
    \caption{(Top) The percentage of light reflected by the glass, depending on its thickness. 
    (Bottom) The theory explaining the plot is illustrated for three values of the thickness: the stopwatches for front and back reflection paths, front and back reflection arrows, and length squares of their sums are shown.} 
    \label{fig:exp2}
\end{figure}
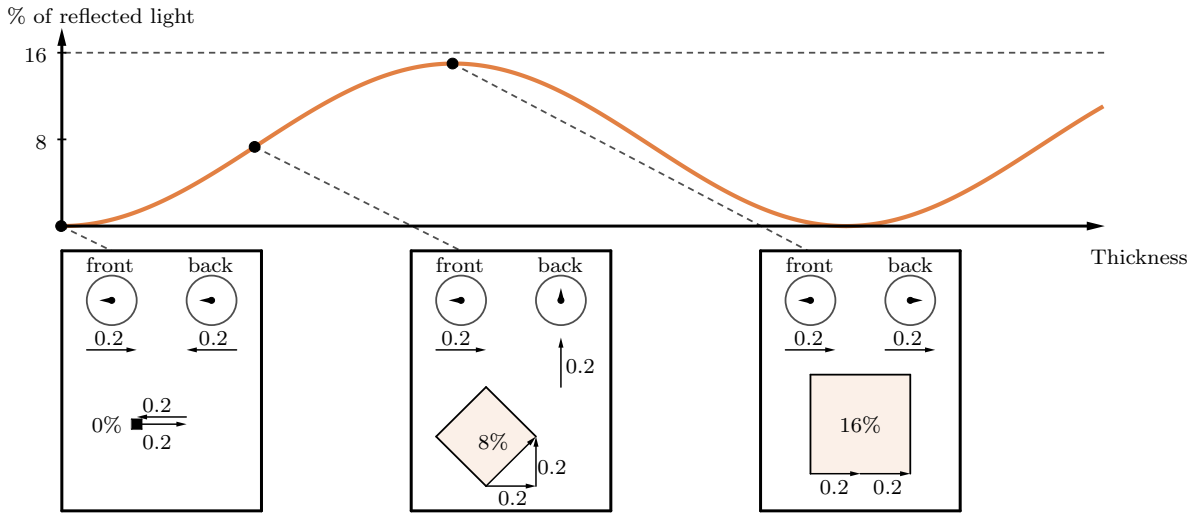

Let us discuss what the experiment says. 
In Figure~\ref{fig:exp2} we plot the percentage of reflected light, depending on the thickness of the glass (see \cite[Figure~5]{Feynman}). 

It turns out that for extremely thin layers of glass, 
light is not reflected at all (or the reflection is very small). This result alone is a powerful blow to both theories suggested above. 
How does a photon hitting the front surface know about the back one, to decide if it should be ever reflected or not? 
The theory of holes and spots completely fails in this case.  

As long as we gradually increase the thickness of the glass, the percentage of the reflected light slowly increases. For some thickness, it reaches the 4\% predicted by the theory of hidden mechanism. But if we increase the thickness of the glass further, the percentage continues to grow, and, at some point, we observe the 8\% predicted by the reflective spots theory. It would seem that we have reached the maximum, but the percentage of reflection increases until it reaches as much as almost 16\%. 
After reaching the maximum, the percentage 
starts decreasing, until it becomes close to zero again for some thickness. Then the same behavior repeats with high accuracy again and again. This unbelievable phenomenon is applied to accurately measuring large distances in a device called \emph{interferometer}. 

One may argue that this experiment just demonstrates the wave nature of light,
and wave optics explains the results. 
This is why we started 
with irrefutable proof that light consists of particles.
This disagrees with wave theory, 
and we need a completely new one. 


\section{Theory.} \label{theory}

\vspace{-0.3cm}

\epigraph{
It is my task to convince you not to turn away because you don't understand it. \dots
Nobody does.}{R.P.~Feynman \cite{Feynman}}

Let us introduce a rough 
theory trying to describe the behavior of the reflection percentage.

This theory will not explain why the behavior is the one we observe. 
It simply predicts the outcome of the experiment. 
We are interested in it because it is a prototype of the entire quantum theory. Something similar happens in 
any quantum model. 
Despite the theory may look paradoxical, it agrees with the experiment. 

The theory cannot predict whether a given photon will be reflected from some surface or not. All that can be predicted is 
the reflection \emph{probability}. 
So, given the thickness of the glass as an input, our theory outputs the reflection probability. This is done in several steps shown in Table~\ref{alg}. Let us explain the calculation from the end to the beginning.
This way we will become more familiar with the basic ideas of quantum theory, starting with the most important ones, and not related to a specific experiment. 

\begin{table}[htb]
\caption{
A recipe for computing the reflection probability for the glass of a given thickness.} 
\label{alg}
\begin{tabular}{ll}
     \hline
     \begin{tabular}{l}{thickness of the glass}
     \end{tabular}
     &  \\
     \multicolumn{1}{c}{$|$} &  \\[-0.32cm]
     \begin{tabular}{l}
     \multicolumn{1}{c}{$\downarrow$}\\
     back/front reflection arrow \\
     \multicolumn{1}{c}{$|$}
     \end{tabular}
     &   
     $\left\{ \begin{tabular}{llp{0.45\textwidth}}
     length & $=$ & $0.2$\\
     direction &$=$ & $+/-$ final direction of the stopwatch hand that makes $15000$ turns per $1$cm of photon's path 
     \end{tabular}\right.$ \\[-0.32cm]
     \multicolumn{1}{c}{$\downarrow$} & \\                        
     \begin{tabular}{l}reflection arrow
     \end{tabular}
     & $=$  \begin{tabular}{l}
     front reflection arrow $+$ back reflection arrow
     \end{tabular} \\      
     \multicolumn{1}{c}{$\downarrow$} &  \\ 
     \begin{tabular}{l}reflection probability
     \end{tabular}
     & $=$ \begin{tabular}{l}
     length square of the reflection arrow 
     \end{tabular}\\
     \hline
\end{tabular}
\end{table}

All we need to do is draw small arrows on paper. 
According to our theory, 
to each event we assign a certain arrow on the plane. The probability of the event is always calculated as the square of the length of the arrow. One can speak of vectors or complex numbers here, but we prefer the word ``arrow''.

Recall that in the first experiment, we had only one reflecting surface and the reflection probability was 4\%. This means that the arrow associated with the reflection event has a length square of $0.04$, hence the length 
is $0.2$. The particular numbers $0.2$ and 4\% here depend on the material. 

Now we have
two surfaces. Let us define the reflection arrow as the sum of the two arrows: the front reflection arrow and the back reflection arrow. So the rules are such that if we have two events that happen in mutually exclusive ways, that is, we are either reflected from the front or the back surface, then we need to sum up the arrows for these events. This rule is somewhat similar to probability theory. 

It remains to compute the back (and front) reflection arrow. 
The length of the arrow for one reflective surface should be $0.2$ so that the square is 4\%. But what is its direction? There are some 
tricky rules at work here. The direction depends on the photon's \emph{path} as follows.
Take an unusual stopwatch with one hand, which rotates at great speed, making approximately 15,000 turns per 
centimeter of the photon's path. (The rotation speed is determined by the color of the light.) When the photon leaves the source, we start this stopwatch. It rotates until the photon reaches the detector. As soon as this happens, we stop our stopwatch. 
The final direction of the stopwatch hand is the desired direction of the back reflection arrow. For the front reflection, the rule is the same with one modification:
we take the \emph{opposite} of the final direction of the 
hand (see \cite[Figure~68d]{Feynman} for an explanation). 
We have finished the theory. 

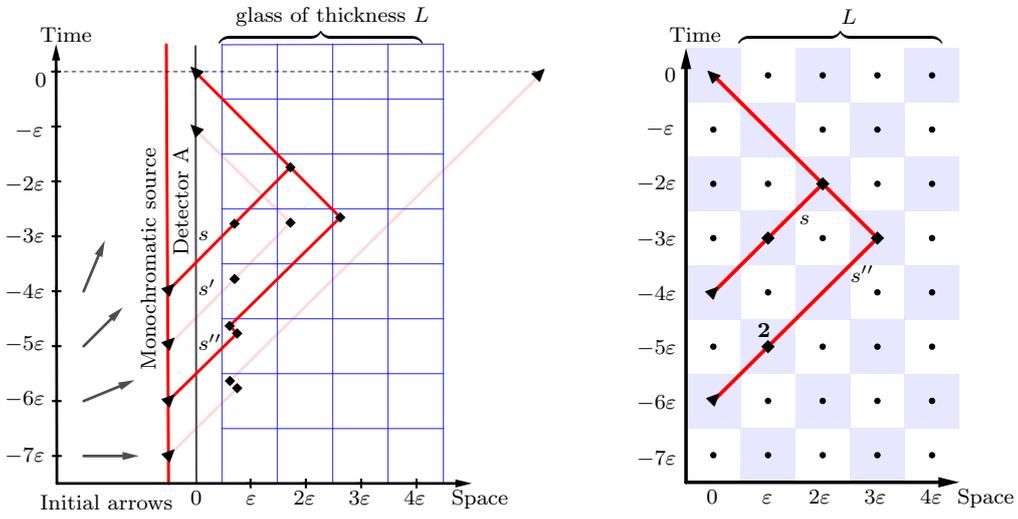
\begin{figure}[t]
\centering
\hspace{7mm}
    \begin{overpic}[width=0.4\linewidth]{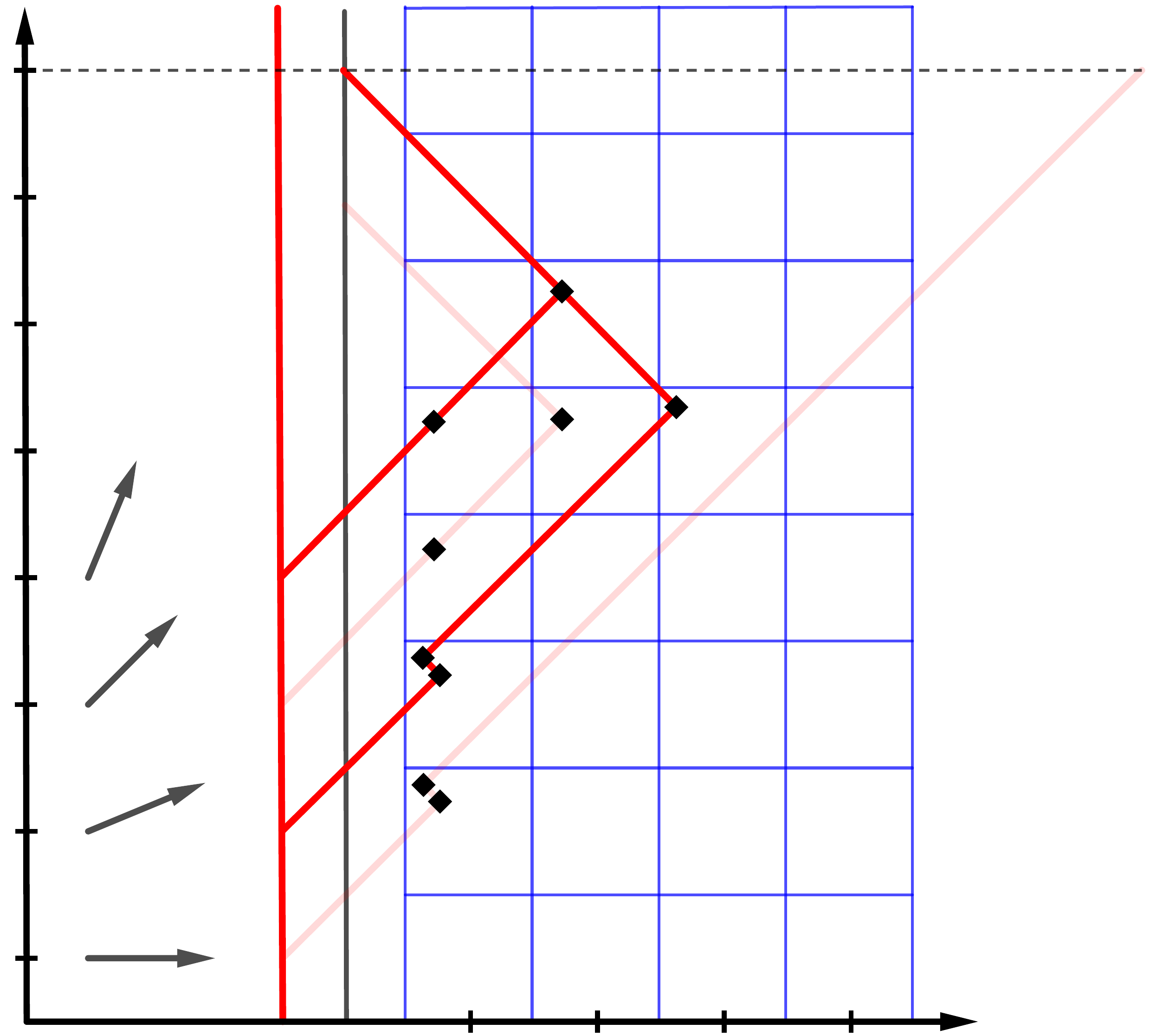}
    \footnotesize{
    \put(-1, 90){Time}
    \put(-2 ,81){$0$}
    \put(-6, 71){$-\varepsilon$}
    \put(-8, 60){$-2\varepsilon$}
    \put(-8, 49.75){$-3\varepsilon$}
    \put(-8, 38.5){$-4\varepsilon$}
    \put(-8, 28){$-5\varepsilon$}
    \put(-8, 17){$-6\varepsilon$}
    \put(-8, 5.5){$-7\varepsilon$}
    \put(-1,-4){\footnotesize{Initial arrows}}
    \put(29, -3){$0$}
    \put(40, -3){$\varepsilon$}%
    \put(49.5, -3){$2\varepsilon$}
    \put(60.5, -3){$3\varepsilon$}
    \put(71.5, -3){$4\varepsilon$}
    \put(81, -3){Space}
    \put(34.75,86.25){$\overbrace{\phantom{0,\varepsilon, 2\varepsilon, 3\varepsilon, 4\varepsilon,5\varepsilon, 6\varepsilon 7}}$}
    \put(38, 94){glass of thickness $L$}
    \put(19.1, 24){{\rotatebox{90}{Monochromatic source}}}
    \put(25.6, 47){{\rotatebox{90}{Detector A}}}
    \put(30.5, 28){$s''$}
    \put(30.5, 39){$s'$}
    \put(30.5, 50){$s$}
    }
    \put(23, 7.5){\rotatebox{-45}{$\blacktriangle$}}
    \put(23, 18.5){\rotatebox{-45}{$\blacktriangle$}}
    \put(23, 29.75){\rotatebox{-45}{$\blacktriangle$}}
    \put(23, 40.75){\rotatebox{-45}{$\blacktriangle$}}
    \put(28, 86){\rotatebox{-135}{$\blacktriangledown$}}
    \put(96.75, 83.25){\rotatebox{135}{$\blacktriangledown$}}
    \put(28, 74.25){\rotatebox{-135}{$\blacktriangledown$}}
    \end{overpic}
    \hspace{1.5cm}
    \begin{overpic}[width=0.24\linewidth]{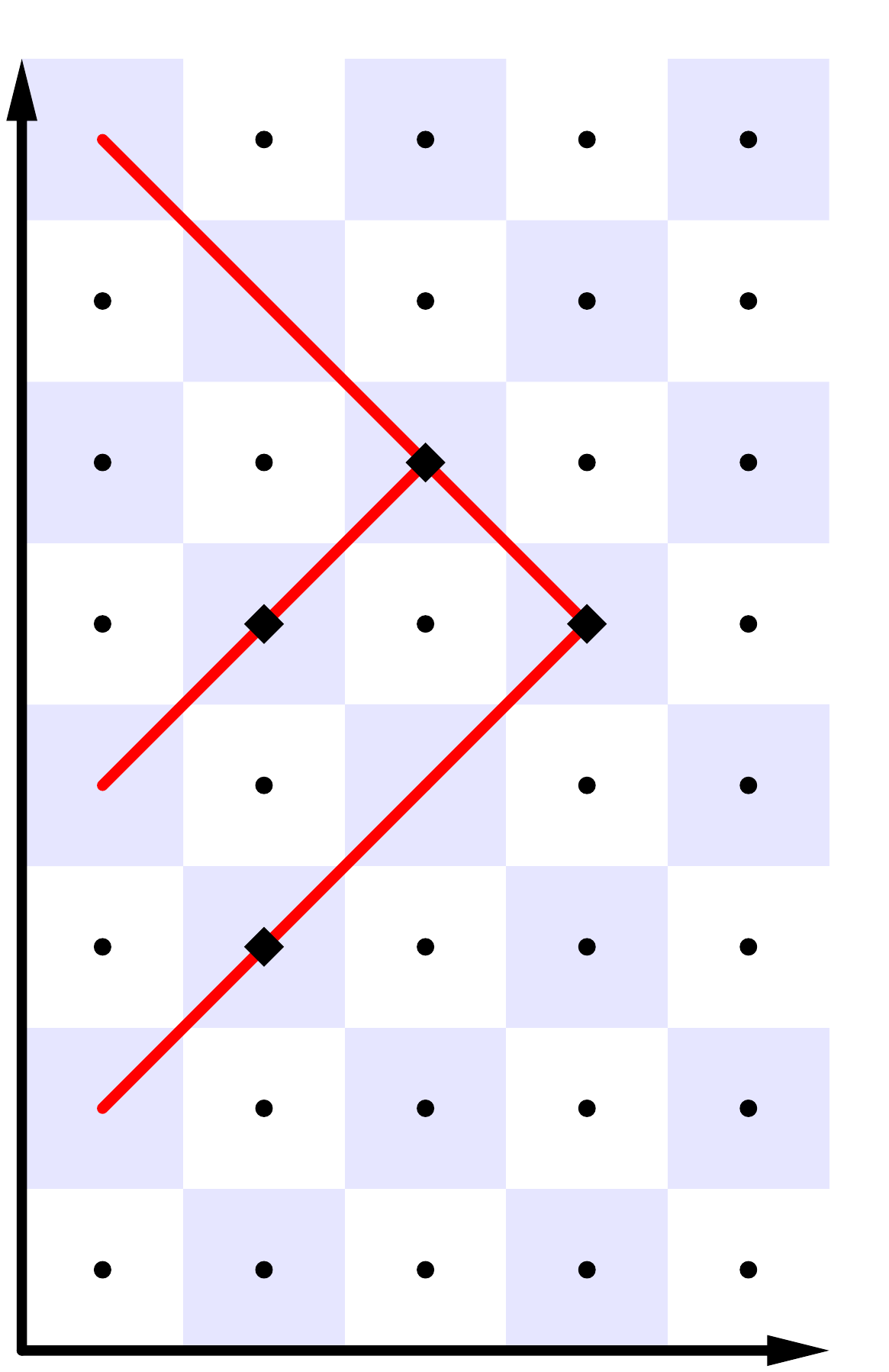}
     \footnotesize{
    \put(-2, 97){Time}
    \put(-3, 88.5){$0$}
    \put(-7, 77){$-\varepsilon$}
    \put(-9, 65){$-2\varepsilon$}
    \put(-9, 53){$-3\varepsilon$}
    \put(-9, 41){$-4\varepsilon$}
    \put(-9, 29.5){$-5\varepsilon$}
    \put(-9, 17.5){$-6\varepsilon$}
    \put(-9, 5.5){$-7\varepsilon$}
    \put(6, -3){$0$}
    \put(18, -3){$\varepsilon$}%
    \put(28, -3){$2\varepsilon$}
    \put(40, -3){$3\varepsilon$}
    \put(52, -3){$4\varepsilon$}
    \put(60, -3){Space}
    \put(37, 45){$s''$}
    \put(26, 57.5){$s$}
    \put(17, 33){$\mathbf{2}$}
    \put(13, 93){$\overbrace{\phantom{0,\varepsilon, 2\varepsilon, 3\varepsilon, 4\varepsilon,5\varepsilon, 6\varepsilon 7}}$}
    \put(35, 101){$L$}
    }
    \put(5.75, 20){\rotatebox{-45}{$\blacktriangle$}}
    \put(5.5, 43.25){\rotatebox{-45}{$\blacktriangle$}}
    \put(5.2, 91.9){\rotatebox{-135}{$\blacktriangledown$}}
    \end{overpic}
    \vspace{3pt}
    \caption{(Left) Reflection as a result of scatterings. The left vertical line (red) depicts a \emph{monochromatic source} 
    that emits photons in a special predictable way. Their initial arrows 
    at four moments are shown to the left from the source.
    The blue thin grid stands for the glass and red thick lines illustrate possible light paths. We 
    use units such that the speed of light is $1$ so that the slope of thick lines is $\pm 1$. Black rhombi depict scatterings. The vertical gray line is detector~A. (Right) A formalization of the left picture. Black points are the centers of the thin blue squares to the left. 
    The bold number ``$\mathbf{2}$'' 
    means 
    two scatterings at the same point. In general, the number of scatterings can be arbitrary. We have added such a possibility 
    to the original Feynman model to make it compatible with classical optics. See Definition~\ref{def-main}. 
    }
    \label{Feynman_pic}
\end{figure}

Let us illustrate it with the example of a very thin layer of glass. See Figure~\ref{fig:exp2} to the bottom-left.
To compute the arrow of the back reflection, start the stopwatch when the photon leaves the source and do 15,000 turns per centimeter of its path. At the end of the path, the stopwatch hand has been fixed at some position. This is the direction of the back reflection arrow,
and the length 
is set to be $0.2$.  
To compute the arrow of the front reflection, run the stopwatch again while the photon is moving. Since our layer is extremely thin, 
the final position of the stopwatch hand 
is almost the same as for the back reflection. 
But this time we need to reverse the direction.
Therefore, 
the reflection arrow 
is the sum of two arrows of the same length and almost opposite directions. 
We get a very small arrow and its length square is even smaller. 
So, our theory indeed 
reproduces the result that the reflection probability is small 
for a very thin layer. 

For thicker glass, the front reflection arrow remains the same but the back one acquires an additional rotation angle, and the probability increases. See Figure~\ref{fig:exp2} to the bottom. 

The reader is cordially invited to play with this theory further and see that it indeed explains the plot in Figure~\ref{fig:exp2}, as well as beautiful phenomena like rainbow patterns on soap bubbles \cite{Feynman}.

%

We need a more accurate model for a more accurate and deeper explanation of partial reflection of light. See Figure~\ref{Feynman_pic} to the left. The modifications are as follows. 

First, the arrow does not really rotate while the photon moves. Instead, the \emph{initial direction} of the arrow periodically depends on the moment when the photon is emitted. The rotation frequency $\omega$, known as the \emph{frequency of emitted light}, is determined by the color and constant for our source.
Photons following paths of different lengths and hitting the detector at the same moment were emitted at distinct moments: thus the rotation of the initial arrow leads to the same result as the rotation during the motion. 

Second, the reflection does not really happen on the surface of the glass; it is 
a result of scattering of light \emph{inside} the glass. A scattering is an absorption of light followed by an immediate emission in an arbitrary direction {(possibly, without changing the direction). Each scattering contracts the arrow 
and rotates it clockwise through the right angle. To describe this process, we decompose the glass of thickness $L$ into layers of small width $\varepsilon$, ignore the coordinates parallel to the layers, and consider time slices with increment $\varepsilon$. This corresponds to decomposing two-dimensional spacetime into squares $\varepsilon\times \varepsilon$. 
A photon can be scattered several times, even in the same square $\varepsilon\times \varepsilon$. 
We have added the possibility of multiple scatterings to the original Feynman model from \cite[Figure~68]{Feynman} to make it consistent with classical optics \cite[\S86]{Landavshits-08}. 

In the following mathematical formalization, we assume that all the scatterings happen only at the centers of the squares $\varepsilon\times \varepsilon$. The resulting light paths resemble moves in the game of checkers, just like in another model introduced by Feynman --- \emph{Feynman checkers} (or \emph{checkerboard}); see \cite{Feynman-Hibbs\anonymize{, SU-22}{}}.


\section{Definition.}
\label{definitions}

\epigraph{
It's like playing checkers: the rules are simple, but you use them over and over.
}
{
R.P. Feynman~\cite{Feynman}
}

It is time for a precise and self-contained definition of the mathematical 
model of thin-film reflection.

\bigskip
\begin{definition} 
\label{def-main}
See Figure~\ref{Feynman_pic} to the right. Fix $\omega,m,L,\varepsilon>0$ called \emph{frequency}, \emph{scattering strength}, \emph{film thickness}, and \emph{lattice step} respectively. 
Take the lattice $\varepsilon\mathbb{Z}^2=\left\{\,(x,t):\frac{x}{\varepsilon},\frac{t}{\varepsilon}\in\mathbb{Z}\,\right\}$. Assume 
$m\varepsilon<1$. 

A \emph{light path} $s$ is a finite sequence of lattice points such that the vector from each point (except the last one) to the next one is cooriented with either $(1,1)$ or $(-1,1)$ and each point (possibly, except the first and the last ones) belongs to the strip $0< x\le L$. 
The vector between any two consecutive points in~$s$ is called a \emph{step}; it can be 
$(0,0)$, $(\pm\varepsilon,\varepsilon)$, $(\pm2\varepsilon,2\varepsilon)$, etc. 
Zero vector is viewed as cooriented with any vector, so 
repeating points in~$s$ 
are allowed. The \emph{number of scatterings} (or \emph{layovers}) in $s$ 
is the number $\ell(s)$ of
points in~$s$ including all repetitions but excluding the first and the last points.

Define the \emph{reflection arrow} (or \emph{reflection amplitude}) to be the infinite sum
\begin{equation}\label{eq-def-model}
   a(\omega,m,L,\varepsilon):=\sum\limits_{\substack{\tau\in \varepsilon\mathbb{Z}\\\tau<0}}\,\,\sum_{s: (0,\tau) \rightsquigarrow (0,0)}
   e^{i\omega \tau}(-im\varepsilon)^{\ell(s)},
\end{equation}   
where the inner summation is over all light paths starting at $(0,\tau)$ and ending at the origin.

Define the \emph{reflection probability} to be
$$P(\omega,m,L,\varepsilon):=|a(\omega,m,L,\varepsilon)|^2.$$
\end{definition}

One still needs to check that the introduced notions are well-defined; that is, both series above converge, and the sums do not depend on the order of summation.
This is postponed until Appendix~\ref{appendix_correct}, and we proceed to the statement of the main result. 

\section{Theorem.}
\label{Theorem}

\epigraph{If you thought that science was certain -- well, that is just an error
on your part.}{R.P. Feynman}

\begin{theorem}[Thin-film reflection probability] \label{th-reflection} 
For each $\omega,m,L>0$ we have
\begin{equation} \label{eq-th-reflection}
\lim_{\varepsilon\searrow 0} P(\omega,m,L,\varepsilon)=
\dfrac{(n^2-1)^2}{(n^2+1)^2+4n^2\cot^2{\omega nL}}, 
\end{equation}
where $n:=\sqrt{1+2m/\omega}$ and the right side is set to be $0$
if $\frac{\omega nL}{\pi}\in\mathbb{Z}$, i.e., the cotangent is undefined.
\end{theorem}


The parameter $n$ has the meaning of \emph{refractive index}. For ordinary glass, $n\approx 1.5$. The probability that light is reflected by a single surface is expressed through the refractive index as $(n-1)^2/(n+1)^2$ \cite[Eq.~(86.9)]{Landavshits-08} 
This indeed equals the value 4\% introduced at the beginning of Section~\ref{experiment}. 

Let us discuss briefly what we can and cannot achieve using this theorem.

Formula~\eqref{eq-th-reflection} is more accurate compared to the receipt in Table~\ref{alg}. According to~\eqref{eq-th-reflection}, the maximal reflection probability is $(n^2-1)^2/(n^2+1)^2\approx 15\%$, which is less than the value of 16\% from Section~\ref{theory}. The two theories agree for small refraction: $(n^2-1)^2/(n^2+1)^2\sim 4(n-1)^2/(n+1)^2$ as $n\to 1$.

The parameter $m$ is not directly measurable in an experiment. What we can do is measure the reflection percentage for a single surface and compute $n$; then~\eqref{eq-th-reflection} gives the percentage for two surfaces.

The model cannot be used to find the 
dependence of 
$n$ on 
$\omega$ because 
$m$ also depends on $\omega$. We view $\omega$ as fixed and thus write $m$ rather than $m(\omega)$, and avoid discussion of the function $m(\omega)$. This would be a very different story related to the question of \emph{why the sky is blue}. See~\cite[Chapters~XV and~IX]{Landavshits-08}.

Known expression~\eqref{eq-th-reflection} for the reflection probability 
is usually derived using wave optics 
\cite[
\S86, Problem~4]{Landavshits-08}. This is not surprising: 
one of the names for quantum mechanics is \emph{wave mechanics}. However, we do know that light consists of particles.
Thus we prefer a combinatorial model and 
proof.


\section{Proof.}
\label{proofs}

\epigraph{If you keep proving stuff that others have done\dots
--- then one day you'll turn around and discover that nobody actually did that one!}{R.P~Feynman}

Our plan is as follows. We are interested in the reflection probability, 
defined in terms of light paths ending at 
the origin. However, to compute the probability, we consider paths ending at an \emph{arbitrary} lattice point $(x,t)$ and find a linear recurrence relation on the sum of their arrows. Identifying the boundary conditions, we get a linear system. Solving it in the limit as $\varepsilon\searrow0$,
we prove the theorem.

This approach is a toy example of using \emph{Smirnov's (para)fermionic observables}, which have made a revolution in statistical physics last $25$ years. \anonymize{See, for instance, \cite{KSS-23} and references therein.}{}


%


\begin{figure}[hb!]
\centering
\begin{tabular}{lccc}
\begin{overpic}[width=0.2\linewidth]{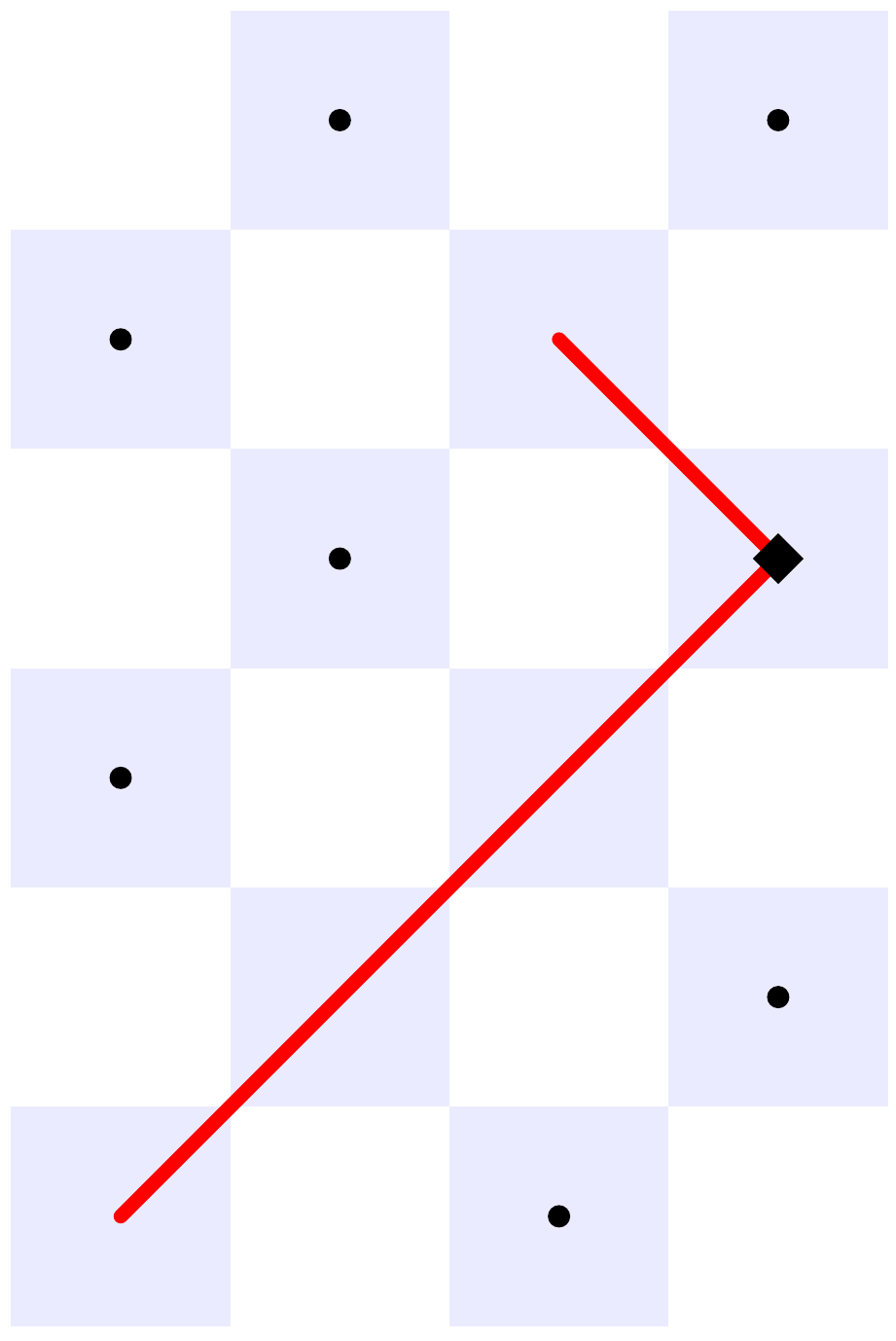}
\put(1,0){\footnotesize{$(0,\tau)$}}
\put(33,78){\footnotesize{$(x,t)$}} 
\put(50, 45){\footnotesize{$s$}}
\put(6,10.25){\rotatebox{-45}{$\blacktriangle$}}
\put(38.5,77.5){\rotatebox{-135}{$\blacktriangledown$}}
\end{overpic}
\hspace{2pt}
& 
\begin{overpic}[width=0.2\linewidth]{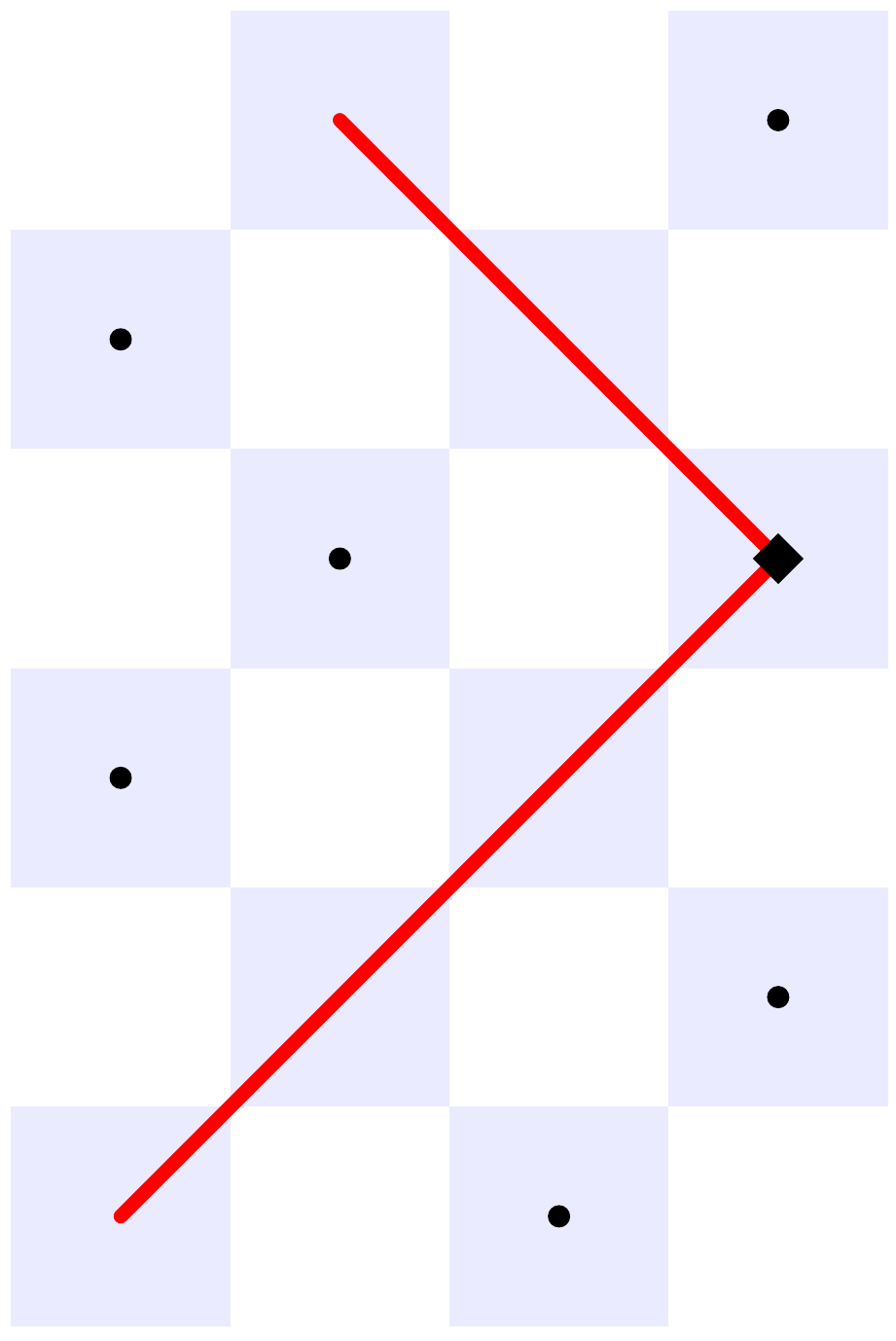} 
\put(1,0){\footnotesize{$(0,\tau)$}}
\put(3,94){\footnotesize{$(x-\varepsilon,t+\varepsilon)$}}  
\put(50, 45){\footnotesize{$s_0$}}
\put(6, 10.25){\rotatebox{-45}{$\blacktriangle$}}
\put(22, 94.25){\rotatebox{-135}{$\blacktriangledown$}}
\end{overpic}
\hspace{2pt}
&
\begin{overpic}[width=0.2\linewidth]{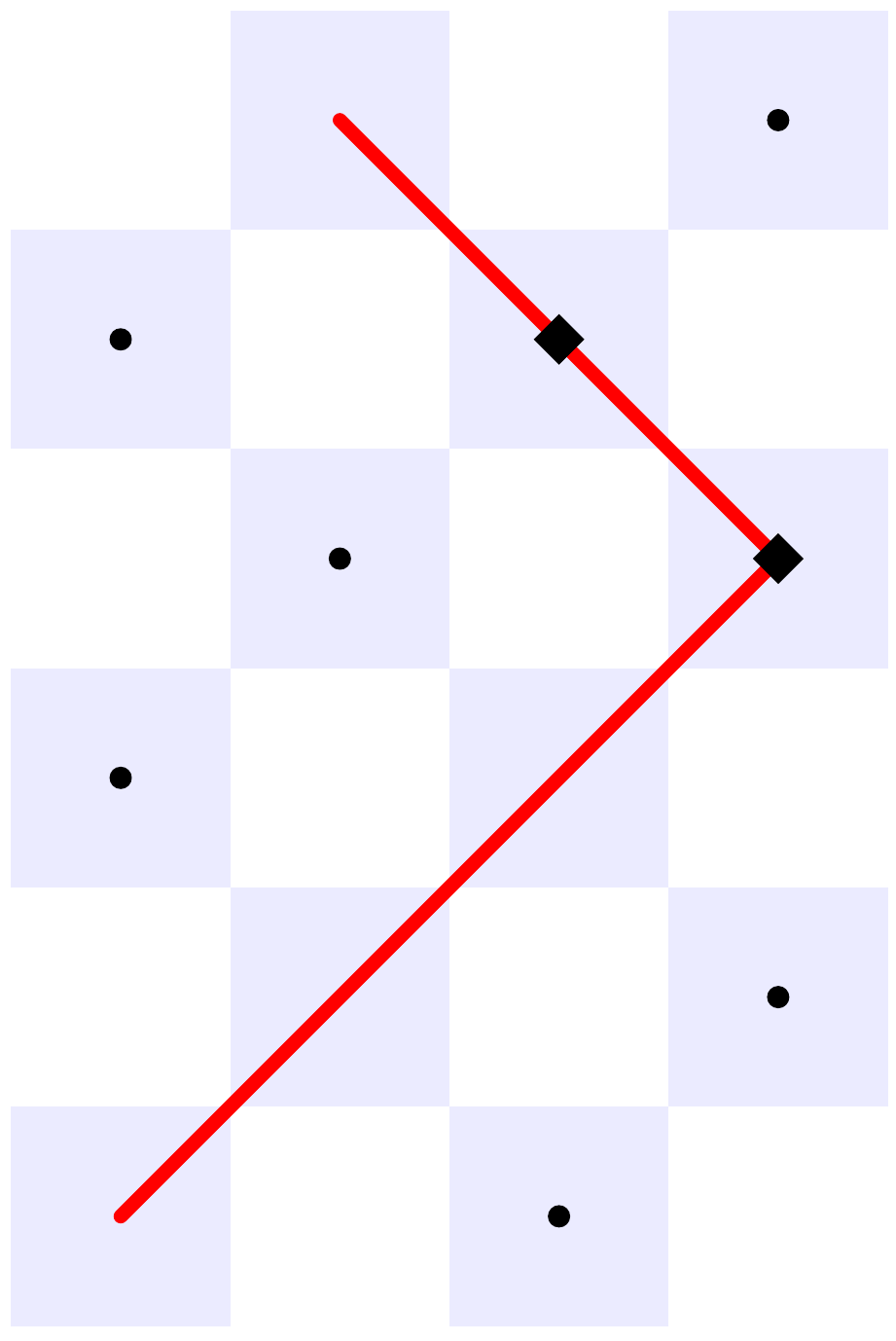} 
\put(1,0){\footnotesize{$(0,\tau)$}}
\put(3,94){\footnotesize{$(x-\varepsilon,t+\varepsilon)$}}  
\put(50, 45){\footnotesize{$s_1$}}
\put(6,10.25){\rotatebox{-45}{$\blacktriangle$}}
\put(22, 94.25){\rotatebox{-135}{$\blacktriangledown$}}
\end{overpic}
\hspace{2pt}
&
\begin{overpic}[width=0.2\linewidth]{DiracS1.pdf} 
\put(1,0){\footnotesize{$(0,\tau)$}}
\put(3,94){\footnotesize{$(x-\varepsilon,t+\varepsilon)$}}
\put(43, 76){\footnotesize{$\mathbf{2}$}}
\put(50, 45){\footnotesize{$s_2$}}
\put(6,10.25){\rotatebox{-45}{$\blacktriangle$}}
\put(22, 94.25){\rotatebox{-135}{$\blacktriangledown$}}
\end{overpic}
\end{tabular}
\caption{A light path $s$ contributing to the value $a_-(x,t)$ and the corresponding paths
$s_0, s_1, s_2, \dots$ contributing to $a_-(x-\varepsilon,t+\varepsilon)$. 
Bold number ``$\mathbf{2}$'' means that the point is repeated twice in the path $s_2$.
See the proof of Lemma \ref{Dirac_1}.} 
\label{Checker-paths-ex}
\end{figure}

So, we introduce the \emph{wave function} in $x$, $t$, and the direction of the last step:
\begin{align}
\label{a+-}
a_{\pm}(x,t) = a_{\pm}(x,t;\omega,m,L,\varepsilon) := \sum_{\substack{\tau\in \varepsilon\mathbb{Z}\\\tau<t}}\sum_{\substack{s: (0,\tau) \rightsquigarrow (x,t)}}
e^{i\omega \tau}(-im\varepsilon)^{\ell(s)},
\end{align}
where the inner summation is over all light paths $s$ starting at $(0, \tau)$ and ending at a lattice point $(x,t)$ with the last step
non-zero and cooriented with $(\pm1, 1)$ (the sign in $\pm$ is the same as in
~\eqref{a+-}).  See Figure~\ref{Checker-paths-ex} to the left. We emphasize that the point $(x,t)$ occurs in $s$ exactly once. Set 
$a(s):=e^{i\omega \tau}(-im\varepsilon)^{\ell(s)}$.

As usual, the convergence issues are postponed until Appendix~\ref{appendix_correct}, and the proofs in this section are valid under the assumption that $a_{\pm}(x,t)$ are well-defined, that is, both series in~\eqref{a+-} converge absolutely. 

Physically, $|a_{+}(x,t)|^2$ and $|a_{-}(x,t)|^2$ for fixed $t$ mean the ``relative probability'' of finding the right- and left-moving photon respectively at the point $x\in [0,L]$ at the time $t$. (To interpret them as genuine probabilities, one needs to truncate and normalize the outer sum in~\eqref{a+-}\anonymize{; 
see \cite[Section~4]{LKTG} for details}{}.)

In what follows, assume that $L/\varepsilon$ is an integer for notational convenience; otherwise, replace $L$ by $\lfloor L/\varepsilon \rfloor \varepsilon$ until the limit $\varepsilon\searrow 0$ is taken. Let us make a few observations.

\begin{example}[Boundary conditions] \label{ex-bc} 
$a_+(\varepsilon,t) = e^{i\omega(t-\varepsilon)}$, $a_-(L,t) = 0$, $a_-(0,0)=a(\omega,m,L,\varepsilon)$.

\begin{proof}
This follows from~\eqref{a+-} because among all light paths $s$ starting on the line $x=0$,
\begin{itemize}
    \item there is a single light path ending at $(\varepsilon,t)$ with the last step non-zero and
cooriented with $(1,1)$;
    \item there are no light paths ending at $(L,t)$ with the last step non-zero and 
cooriented with $(-1,1)$;
    \item there are no light paths ending at $(0,0)$ with the last step 
non-zero and cooriented with $(1,1)$.
\end{itemize}
\vspace{-0.6cm}
\end{proof}
%
%
\end{example}






\begin{lemma}[Reccurence]
\label{Dirac_1}
For each lattice point $(x,t)$ with $0<x\le L$,
the following equalities hold:
\begin{equation}
 \label{eq-Dirac}
\begin{aligned}
    a_-(x-\varepsilon,t+\varepsilon)
&=\frac{1}{1+im\varepsilon}a_-(x,t)
+\frac{-im\varepsilon}{1+im\varepsilon}a_+(x,t);\\
a_+(x+\varepsilon,t+\varepsilon)
&=\frac{-im\varepsilon}{1+im\varepsilon}a_-(x,t)
+\frac{1}{1+im\varepsilon}a_+(x,t).
\end{aligned}
\end{equation}
\end{lemma}

\begin{proof}
    Let us prove the first equality; the second one is proved analogously. 
    To each light path $s$ contributing to the right side of the equality we assign an infinite collection of light paths contributing to the left side and show that the contributions coincide. See Figure~\ref{Checker-paths-ex}. Consider two cases, depending on whether $s$ contributes to $a_-(x,t)$ or to $a_+(x,t)$.
     
    \emph{Case 1:} $s$ contributes to $a_+(x,t)$. This means that $s$ is a light path from $(0,\tau)$ to $(x,t)$, and the last step is non-zero and cooriented with $(1,1)$. 
    Recall that a light path is a sequence of points of the lattice. For $j=1,2,\dots$, denote by 
    $s_j$ the light path obtained by appending $j-1$ times the point $(x,t)$ and one time the point $(x-\varepsilon, t+\varepsilon)$ at the end of $s$. The collection of light paths $s_1,s_2,\dots$ is the desired one. By definition, $a(s_j)= (-im\varepsilon)^j a(s)$ so that $\sum_{j=1}^\infty a(s_j)=\frac{-im\varepsilon}{1+im\varepsilon}a(s)$ as geometric series.
     
    \emph{Case 2:} $s$ contributes to $a_-(x,t)$. This means that $s$ is a light path from $(0, \tau)$ to $(x,t)$, and the last step 
    is non-zero and cooriented with $(-1,1)$. For $j=1,2,\dots$, denote by $s_j$ the light path obtained by appending $j-1$ times the point $(x,t)$ and one time the point $(x-\varepsilon, t+\varepsilon)$ at the end of~$s$. For $j=0$, denote by $s_0$ the light path obtained by replacing the point $(x,t)$ with $(x-\varepsilon, t+\varepsilon)$. The collection of light paths $s_j$ is the desired one: $a(s_j)= (-im\varepsilon)^j a(s)$ so that $\sum_{j=0}^\infty a(s_j)=\frac{1}{1+im\varepsilon}a(s)$ as geometric series.

All the paths contributing to the left side have been taken into account. Indeed, if a light path $s'$ contributes to the left side, then removing the last point and leaving only one instance of point $(x,t)$ (adding it if it is absent) we get a light path $s$ contributing to the right side
so that $s'=s_j$ for some~$j$.

We have shown that the contribution of each light path $s$ to the right side and the family of paths $s_j$ to the left side of the first equality in~\eqref{eq-Dirac} coincide. Since the inner series in~\eqref{a+-} is absolutely convergent (see Lemma~\ref{inner_conv} below) 
and thus independent of summation order, it follows that the right and left sides of the first equality are equal.
\end{proof}



\begin{lemma}[Quasiperiodicity] \label{l-quasiperiodicity} 
For each lattice point $(x,t)$, 
we have  $a_{\pm}(x,t) = e^{i\omega t}a_{\pm}(x,0)$. 
\end{lemma}

\begin{proof}
   To each light path $s'$ from $(0,\tau)$ to $(x,t)$, contributing to \eqref{a+-}, assign the light path $s$ from $(0,\tau-t)$ to $(x,0)$ obtained by a shift in the $t$-direction. See Figure \ref{Feynman_pic} to the left.
   We get $a(s') = e^{i\omega t}a(s)$. Summing over all $s'$ and $\tau<t$, we get the desired 
   result.
\end{proof}


\begin{proof}[Proof of Theorem \ref{th-reflection}] 
Lemmas~\ref{Dirac_1} and~\ref{l-quasiperiodicity} give us the following system of equations for $x = \varepsilon, 2\varepsilon,\dots, L$:
\begin{align}
    \label{diracT0a-}
         a_-(x-\varepsilon,0)e^{i \omega \varepsilon}
         &=\frac{1}{1+im\varepsilon}a_-(x,0)
         +\frac{-im\varepsilon}{1+im\varepsilon}a_+(x,0);\\
    \label{diracT0a+}
         a_+(x+\varepsilon,0)e^{i \omega \varepsilon}
         &=\frac{-im\varepsilon}{1+im\varepsilon}a_-(x,0)
         +\frac{1}{1+im\varepsilon}a_+(x,0).
\end{align}
Example~\ref{ex-bc} gives the following boundary conditions:
\begin{align}
\label{Bound_cond}
a_+(\varepsilon,0) = e^{-i\omega\varepsilon} \qquad\text{and}\qquad
a_-(L,0) = 0.
\end{align}
We are going to solve system~\eqref{diracT0a-}--\eqref{Bound_cond} in the limit $\varepsilon\searrow 0$ for fixed $\omega$ and $m$, and thus find the reflection probability $\lim\limits_{\varepsilon\searrow 0}|a_-(0,0)|^2$. 

First let us prove that for small enough $\varepsilon>0$, any solution of system~\eqref{diracT0a-}--\eqref{diracT0a+} has the form 
\begin{align}
\label{eq_1}
    &a_+(x,0) = a(\varepsilon) e^{ik(\varepsilon) x} + b(\varepsilon)e^{-ik(\varepsilon) x}, \\
\label{eq_2}
    &a_-(x,0) = c(\varepsilon) e^{ik(\varepsilon) x} + d(\varepsilon)e^{-ik(\varepsilon) x},
\end{align}
for some coefficients $a(\varepsilon), b(\varepsilon), c(\varepsilon), d(\varepsilon) \in \mathbb{C}$ and  $k(\varepsilon)>0$ depending on $\varepsilon$. Such form has a clear physical meaning: light inside the medium is the superposition of the reflected wave and the transmitted wave, with a new wavelength  $\lambda(\varepsilon):={2 \pi}/{k(\varepsilon)}$. 

Indeed, expressing $a_{-}(x,0)$ from \eqref{diracT0a+}, substituting into \eqref{diracT0a-}, and simplifying, 
we obtain the following recurrence relation for $x = 2\varepsilon, 3\varepsilon, \dots, L-\varepsilon$: \begin{align}
\label{KG_rec}
    a_{+}(x-\varepsilon,0) - 2(\cos \omega\varepsilon - m\varepsilon \sin \omega\varepsilon)a_{+}(x,0) + a_{+}(x+\varepsilon,0) = 0.
\end{align}
The function $a_{-}(x,0)$ satisfies the same relation.
It is well-known that a general solution of the recurrence relation has form~\eqref{eq_1}, where $e^{\pm ik(\varepsilon) \varepsilon}$ are the two roots of the \emph{characteristic polynomial}
\begin{align}
\label{P(lambda)}
    P(\chi)=\chi^2 - 2(\cos \omega\varepsilon - m\varepsilon \sin \omega\varepsilon)\chi + 1. 
\end{align}
Since 
$$|\cos \omega\varepsilon - m\varepsilon \sin \omega\varepsilon| < 1$$
for small enough $\varepsilon$, it follows that $k(\varepsilon)$ is a real number. Assume WLOG $k(\varepsilon)>0$. 

Now let us find the limits of the coefficients $a(\varepsilon), b(\varepsilon), c(\varepsilon), d(\varepsilon),k(\varepsilon)$ and the reflection probability $|a_-(0,0)|^2=|c(\varepsilon) + d(\varepsilon)|^2$ as $\varepsilon\searrow 0$. 

By Vieta's formulae, it follows that
\begin{align*}
    \cos k(\varepsilon)\varepsilon = \cos \omega\varepsilon - m\varepsilon \sin \omega\varepsilon = 1 - \left(m\omega+\frac{\omega^2}{2}\right)\varepsilon^2 + o(\varepsilon^2)\qquad\text{as }\varepsilon \searrow 0.
\end{align*}
Passing to the limit, 
we obtain
\begin{align}\label{eq-k}
    k := \lim\limits_{\varepsilon \searrow 0}k(\varepsilon)=\omega \sqrt{1+\frac{2m}{\omega}}=\omega n. 
\end{align}
This relation is a well-known result from optics: the wavelength $\lambda:=2\pi/k$ 
in the medium is smaller than the wavelength $\lambda_{\text{vac}}:=2\pi/\omega$ in the vacuum, and $\lambda=\lambda_{\text{vac}}/n$.

Substituting \eqref{eq_1}--\eqref{eq_2} into \eqref{diracT0a+}--\eqref{Bound_cond} and comparing the coefficients at $e^{\pm ik(\varepsilon)x}$, we get the following linear system in the variables $a(\varepsilon), b(\varepsilon), c(\varepsilon), d(\varepsilon)$:
\begin{equation}\label{InitSys}
\begin{aligned} 
c(\varepsilon)&=a(\varepsilon)\cdot \frac{ 1 - e^{i(\omega+k(\varepsilon))\varepsilon}(1+im\varepsilon)}{im\varepsilon},
&
a(\varepsilon)e^{ik(\varepsilon)\varepsilon\,} + b(\varepsilon)e^{-ik(\varepsilon)\varepsilon\,}&=e^{-i\omega\varepsilon},
\\
d(\varepsilon)&=b(\varepsilon)\cdot \frac{1  - e^{i(\omega - k(\varepsilon))\varepsilon }(1+im\varepsilon)}{im\varepsilon},
&
c(\varepsilon)e^{ik(\varepsilon)L} + d(\varepsilon)e^{-ik(\varepsilon)L}&=0. 
\end{aligned}
\end{equation}
%
%
%
Passing to the limit as $\varepsilon \searrow 0$ we obtain 
\begin{equation}\label{limSys}
    \begin{aligned}
    c &= a \cdot \frac{-m-\omega-k}{m},  & a + b&=1,\\
    d &= b \cdot \frac{-m-\omega+k}{m},  & c e^{ikL} + d e^{-ikL} &= 0,
    \end{aligned}
\end{equation}
where $a,b,c,d$ are the limits of $a(\varepsilon), b(\varepsilon), c(\varepsilon), d(\varepsilon)$. 
Since linear system \eqref{limSys} has clearly a unique solution $(a,b,c,d)$ for any $\omega,m,L,k>0$ and the coefficients of linear system~\eqref{InitSys} depend on $\varepsilon$ continuously, it follows that 
the limit of the solution $(a(\varepsilon), b(\varepsilon), c(\varepsilon), d(\varepsilon))$ 
exists and equals $(a,b,c,d)$. 

A direct computation using~\eqref{eq-k} gives 
$$
c + d = -\frac{m}{m+\omega-ik\cot kL},
$$
where the right side 
is set to be $0$
when $\frac{kL}{\pi}\in\mathbb{Z}$, i.e., the cotangent is undefined.
Hence
\begin{align*}
    \lim\limits_{\varepsilon\searrow 0} P(\omega,m,L,\varepsilon) = 
    |c+d|^2
    =\frac{(2m/\omega)^2}{(2m/\omega+2)^2 + 4(k/\omega)^2\cot^2 kL}
    = \frac{(n^2-1)^2}{(n^2+1)^2 + 4n^2 \cot^2 n\omega L}.
\end{align*}
\end{proof}

\section{Problems.}
\label{Problems}

\epigraph{And, so I dreamed that if I were clever, I would find a formula for the amplitude of a path that was beautiful and simple for three dimensions of space and one of time,
which would be equivalent to the Dirac equation\dots
}{R.P. Feynman, Nobel lecture}

Let us discuss a few equivalent definitions, related models, and open problems.

Take a closer look at the \emph{inner} sum in~\eqref{eq-def-model} and~\eqref{a+-}, which is more fundamental, and for $\tau< t$ denote 
\begin{align}
\label{a(T)+-}
a_{\pm}(x,t;\tau) = a_{\pm}(x,t;m,L,\varepsilon, \tau) := \sum_{\substack{s: (0,\tau) \rightsquigarrow (x,t)}}(-im\varepsilon)^{\ell(s)},
\end{align}
where the summation is over the same light paths $s$ in the strip $0<x\le L$ as in~\eqref{a+-}. Finding an asymptotic formula for expression~\eqref{a(T)+-} as $\varepsilon\searrow 0$ while the other parameters are fixed is an open problem.

Here it is easy to get rid of repeating points in the light paths. Indeed, define a \emph{checker path} $p$ to be a light path such that each step
equals either $(\varepsilon,\varepsilon)$ or $(-\varepsilon,\varepsilon)$  (not just is cooriented with those vectors). Then $\ell(p)=(t-\tau)/\varepsilon-1$. Let $\mathrm{turns}(p)$ be the number of \emph{turns}, i.e. pairs of orthogonal consecutive steps.
Summing geometric series at each point (see 
Figure~\ref{Checker-paths-ex} except the leftmost picture and the proof of Lemma~\ref{inner_conv} below), one can rewrite $a_{\pm}(x,t;\tau)$ as a sum over checker paths $p$ from $(0,\tau)$ to $(x,t)$:
\begin{align}
\label{a(T)+-p}
a_{\pm}(x,t;\tau) = \sum_{\substack{p: (0,\tau) \rightsquigarrow (x,t)}}
\frac{(-im\varepsilon)^{\mathrm{turns}(p)}}{(1+im\varepsilon)^{\ell(p)}}.
\end{align}

The resulting model is known as \emph{Feynman checkers} or a \emph{quantum walk}. 
Physically, $|a_{+}(x,t;\tau)|^2$ and $|a_{-}(x,t;\tau)|^2$ are interpreted as the probability of finding the right- and left-moving \emph{electron} respectively at the point $x\in [0,L]$ at the time $t$ if it was emitted from the point $0$ at the time $\tau<0$. An electron is a particle very different from a photon but unexpected relations between them do happen. The parameter $m$ is interpreted as the electron \emph{mass} (in the units where Plank's constant is $1$) and~\eqref{eq-Dirac} is a lattice analog of \emph{Dirac equation} describing the motion of an electron along a line. Usually, denominators in~\eqref{a(T)+-p} and~\eqref{eq-Dirac} contain $\sqrt{1+m^2\varepsilon^2}$ instead of $1+im\varepsilon$; this does not affect the probabilities because 
the two numbers have the same absolute value and $\ell(p)$ is the same for all checker paths $p$ in~\eqref{a(T)+-p}. So, the more conventional definition is  
\begin{align}
\label{a+-p}
a_{\pm}(x,t;m,\varepsilon) = \sum_{\substack{p: (0,0) \rightsquigarrow (x,t)}}
\frac{(-im\varepsilon)^{\mathrm{turns}(p)}}{(1+m^2\varepsilon^2)^{\ell(p)/2}},
\end{align}
where the checker paths $p$ are usually \emph{not} restricted to the strip $0<x\le L$
(such a restriction would mean the electron is absorbed at the points $x=0$ and $x=L+\varepsilon$).

Let us look at this from the computational point of view. 
Running a quantum walk many times, we can ``compute'' the probability distribution $|a_{\pm}(x,t;\tau)|^2$ approximately. If we modify the system so that the probability distribution becomes a desired target function, then we get a \emph{quantum algorithm} for computing the function. \emph{Quantum walks on graphs} have ultimate computational power: 
they can solve any problem that any other quantum computer can~\cite{Venegas-Andraca-12}. 

Quantum walks 
are popular 
in physics. 
See\anonymize{~\cite{LKTG} for an elementary introduction 
and}{} \cite{Kempe-09, Kuyanov-Slizkov-22, Novikov-22, \anonymize{SU-22,}{} Venegas-Andraca-12} for surveys and open 
problems.
\anonymize{Our 
work realizes one step of the research program outlined in~\cite[\S7]{SU-22}.}{} 

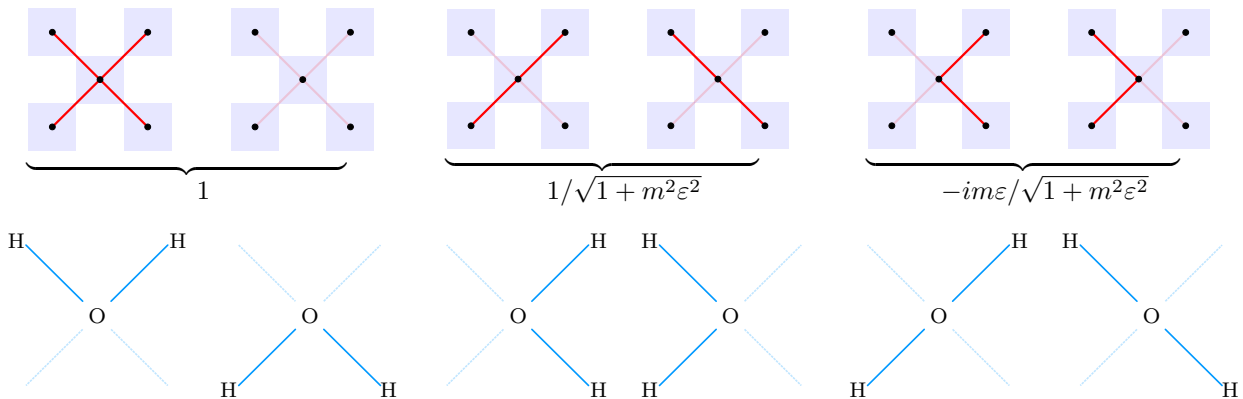
\begin{figure}[htbp]
\centering
  \begin{tabular}{cccccc}
  \multicolumn{2}{c}{
  \begin{overpic}[width=0.31\textwidth]{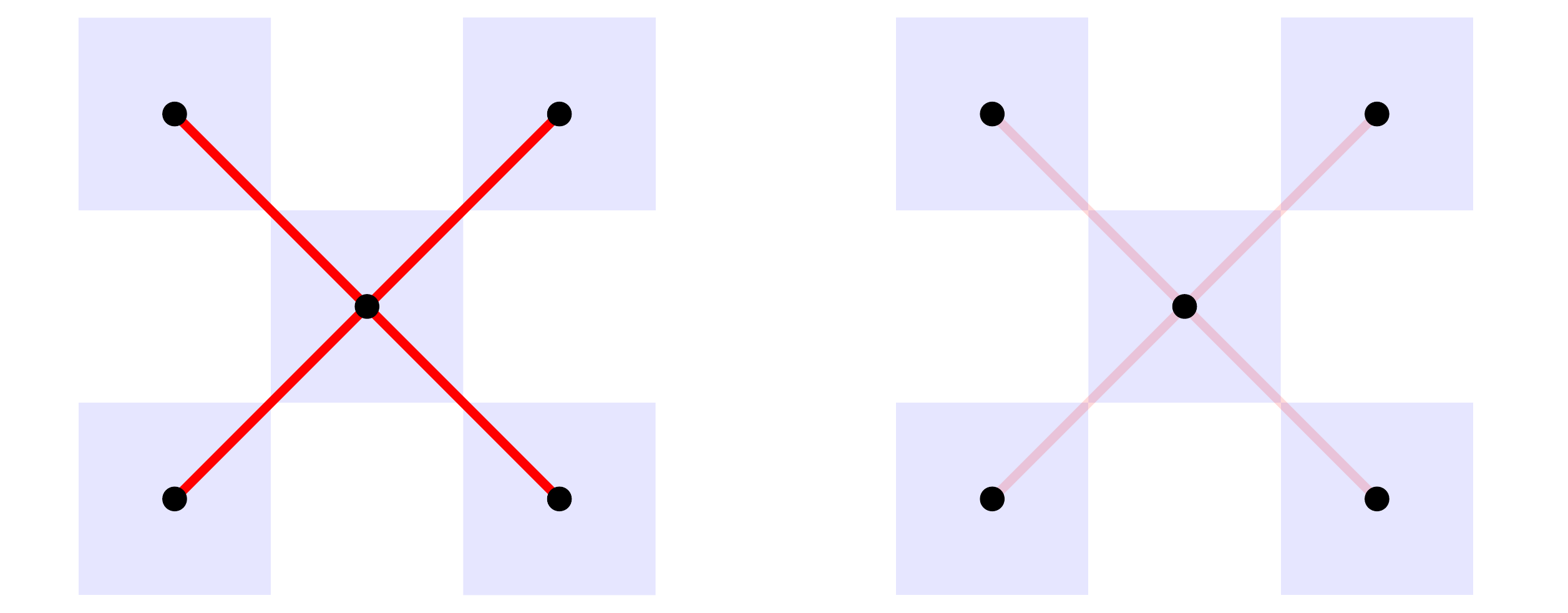}
  \footnotesize{
  \put(4.5,0){$\underbrace{\phantom{aaaaaaaaaaaaaaaaaaaaaaaaaa,}}$}
  }
  \end{overpic}
  }&
  \multicolumn{2}{c}{
  \begin{overpic}[width=0.31\textwidth]{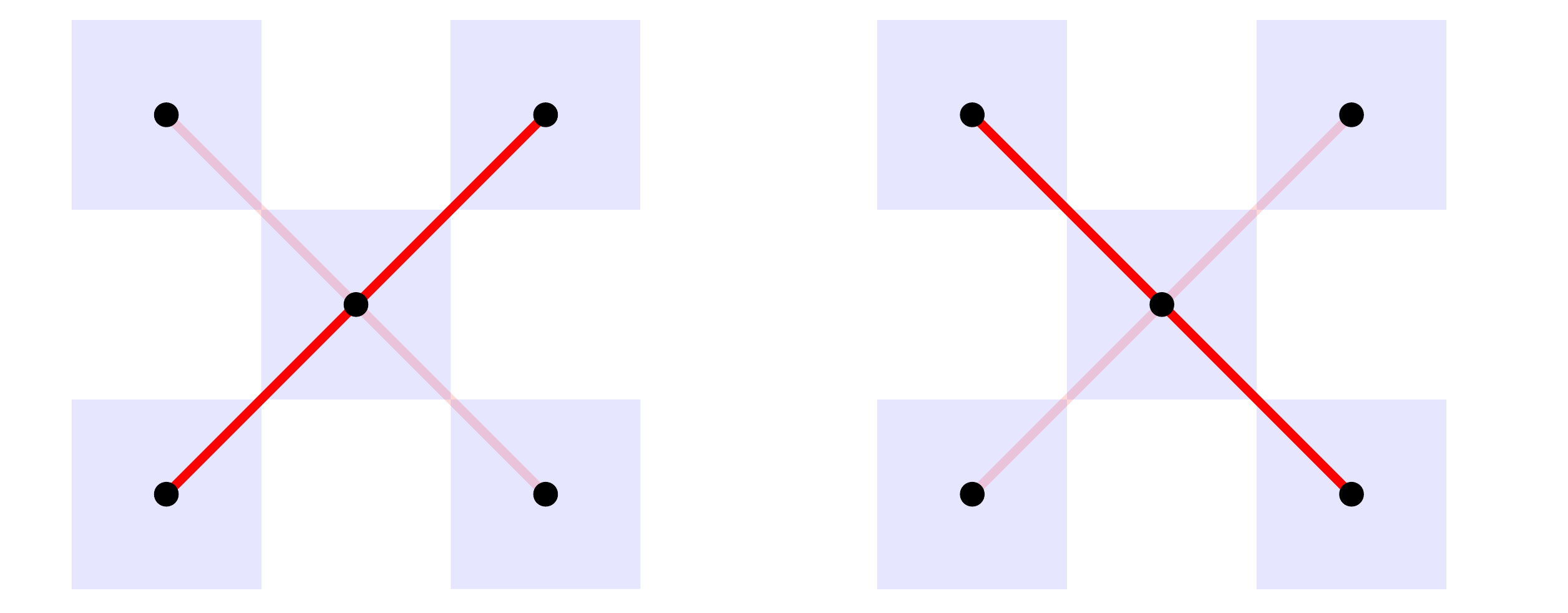}
  \footnotesize{
 \put(4,0){$\underbrace{\phantom{aaaaaaaaaaaaaaaaaaaaaaaaaa}}$}
  }
  \end{overpic}
  }&
  \multicolumn{2}{c}{
  \begin{overpic}[width=0.31\textwidth]{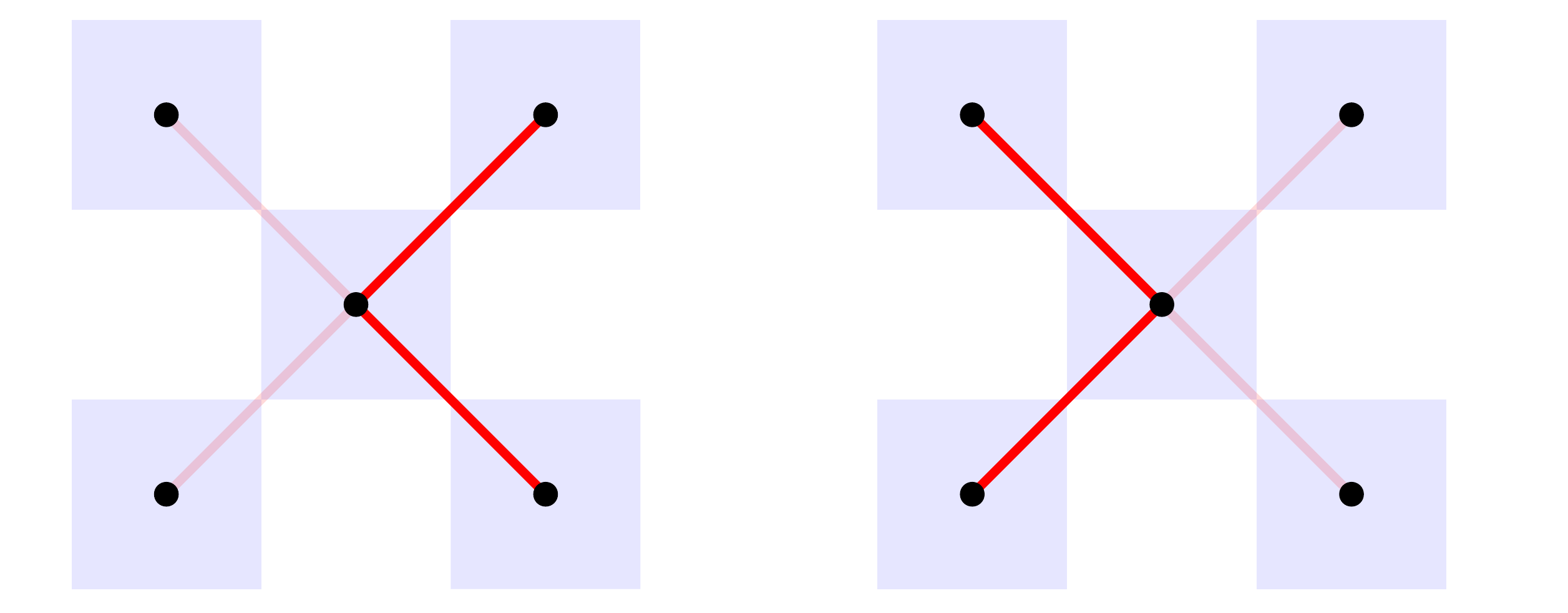}
  \footnotesize{
  \put(4,0){$\underbrace{\phantom{aaaaaaaaaaaaaaaaaaaaaaaaaa}}$}
  }
  \end{overpic}
  }\vspace{4pt}\\
  \multicolumn{2}{c}{$1$} &
  \multicolumn{2}{c}{$1/\sqrt{1+m^2\varepsilon^2}$} & 
  \multicolumn{2}{c}{$-im\varepsilon/\sqrt{1+m^2\varepsilon^2}$}\\
  &&&&&\\
  \multicolumn{2}{c}{
  \begin{overpic}[width=0.3\textwidth]{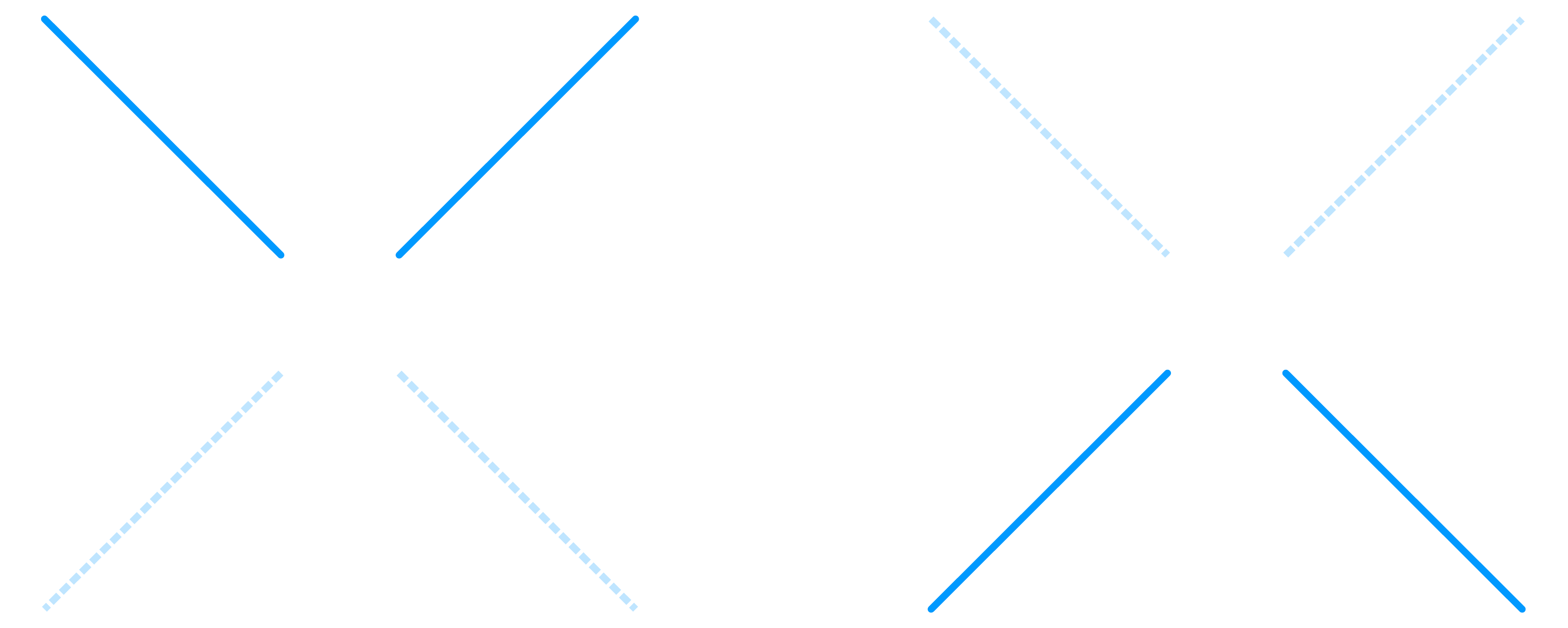}
  \footnotesize{
  \put(19.5,18){O}
  \put(-2, 38){H}
  \put(41, 38){H}
  \put(76, 18){O}
  \put(54.5, -2){H}
  \put(97, -2){H}
  }
  \end{overpic}
  }&
  \multicolumn{2}{c}{
  \begin{overpic}[width=0.3\textwidth]{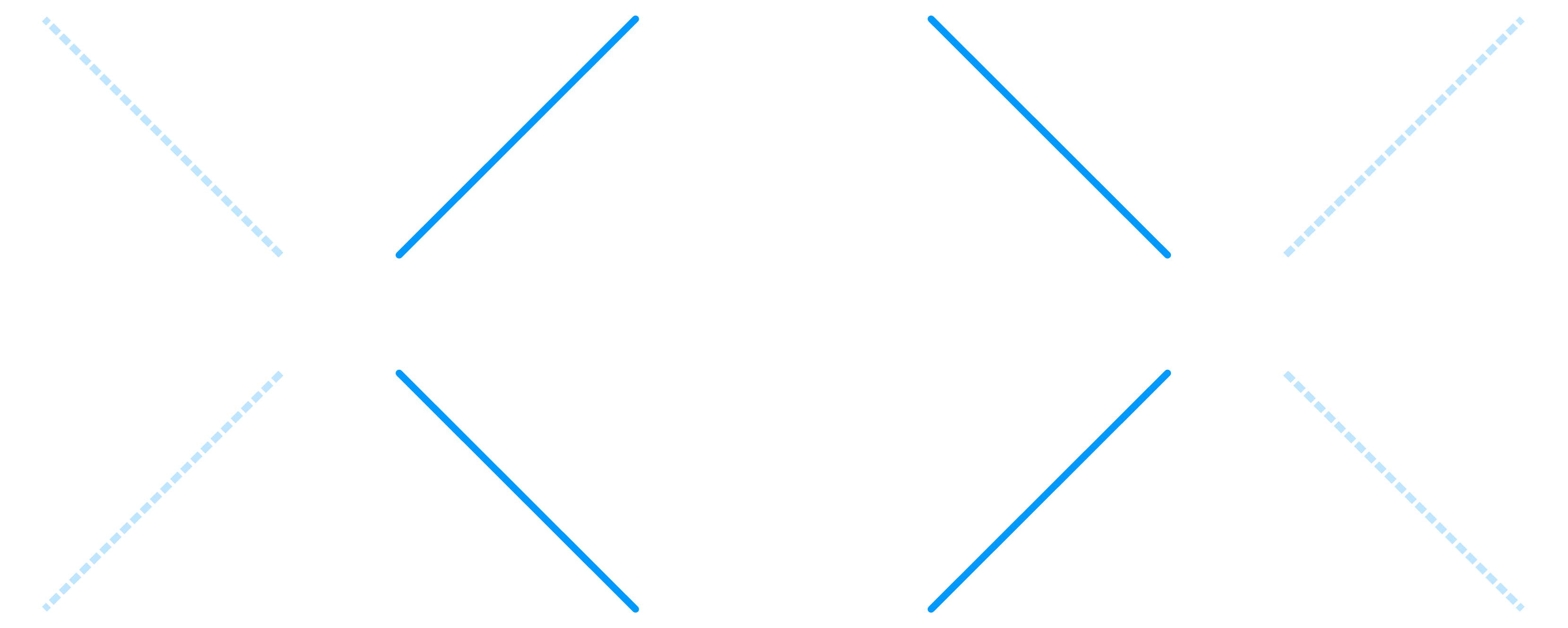}
  \footnotesize{
  \put(19.5,18){O}
  \put(41, -2){H}
  \put(41, 38){H}
  \put(76, 18){O}
  \put(54.5, -2){H}
  \put(54.5, 38){H}
  }
  \end{overpic}
  }&
  \multicolumn{2}{c}{
  \begin{overpic}[width=0.3\textwidth]{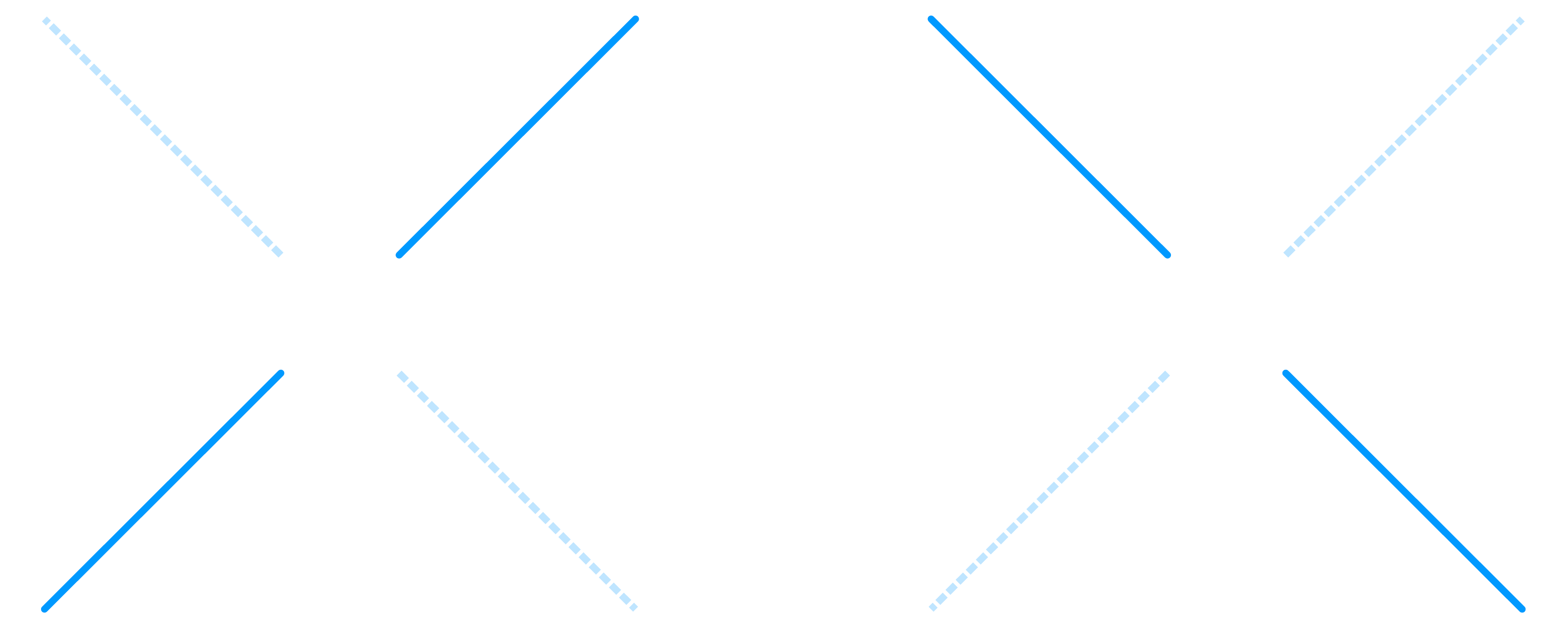}
  \footnotesize{
  \put(19.5,18){O}
  \put(41, 38){H}
  \put(-2, -2){H}
  \put(76, 18){O}
  \put(97, -2){H}
  \put(54.5, 38){H}
  }
  \end{overpic}
  }
  \end{tabular}
  \caption{(Top) Feynman checkers as a six-vertex model. The steps 
  of a checker path are thick red segments; 
  the remaining segments joining diagonal neighbors are thin light.
  Around each lattice point, they form one of the six configurations. The leftmost configuration does not appear but becomes possible in a more general model describing several moving electrons. We take the product of weights in the middle row over all lattice points. (Bottom) 
  The corresponding configurations of ice molecules at a node of the square crystal lattice.
  }\label{fig-6v}
\end{figure}

Yet another equivalent definition is obtained as follows. 
Given a checker path $p$, paint red all the steps.
The segments emanating from each lattice point (except for the starting point and the endpoint of the path) form one of the six configurations shown in Figure~\ref{fig-6v}. Assign complex weights to the configurations as in the middle row of the figure. Then each summand in~\eqref{a+-p} can be computed as the product of weights of all lattice points, and the sum can be taken over all configurations of red segments. 

The resulting model is a particular case of the famous \emph{six-vertex model} also known as \emph{ice model} 
(see Figure~\ref{fig-6v} to the bottom). It is studied for various (usually, real) weights of vertices and various boundary conditions;  see, e.g.,~\cite{Duminil-Copin-etal-22}. In general, the limit as $\varepsilon\searrow 0$ is a widely open problem.

\appendix
\section{Technicalities.}
\label{appendix_correct}

\epigraph{All that might seem to you to be a waste of time --- some silly game for mathematicians only.}{R.P. Feynman~\cite{Feynman}}


We aimed to make the introduction to quantum theory as simple as possible and to avoid any technicalities as long as we could. Now it is the time 
to demonstrate yet another feature of quantum theory:
a rigorous justification of computation, especially convergence issues, often requires more effort than the computation itself\anonymize{ \cite{SU-22}}{}. We thus have to use less elementary mathematical tools in this appendix.

\begin{theorem}[The wave function is well-defined] \label{correctness}
Both series in~(\ref{a+-}) converge absolutely.
\end{theorem}


\begin{lemma} \label{inner_conv} 
The inner series in \eqref{a+-} converge absolutely for $m\varepsilon<1$, that is, $a_{\pm}(x,t;\tau)$ is well-defined via~\eqref{a(T)+-}.
\end{lemma}

\begin{proof}
Let $s$ be a light path from $(0,\tau)$ to $(x,t)$. Joining its consecutive points, we get a broken line (see Figure~\ref{Checker-paths-ex}). To the path $s$, assign the checker path $p=p(s)$ with the same broken line; thus $p$ is obtained from $s$ by removing the repetitions and subdividing the steps. Let $n_i$ be the number of times the $(i+1)$-th point of $p$ is repeated in $s$. If the $(i+1)$-th point is a turn then $n_i\ge 1$. Other points can have arbitrary multiplicities $n_i\ge 0.$ 
Therefore for any fixed checker path $p$ from $(0,\tau)$ to $(x,t)$,
\begin{align*}
    \sum_{\substack{s: p(s)=p}}\left|-im\varepsilon\right|^{\ell(s)}
    =\sum_{\substack{s: p(s)=p}}
    (m\varepsilon)^{n_{1}+ \cdots+ n_{\ell(p)}}
    =\sum_{n_1}(m\varepsilon)^{n_{1}}\cdot\dots\cdot
    \sum_{n_{\ell(p)}}(m\varepsilon)^{n_{\ell(p)}}
    =\frac{(m\varepsilon)^{\mathrm{turns}(p)}}{ (1 - m\varepsilon)^{\ell(p)}},
\end{align*}
where we sum over $n_i\ge 1$, if the $(i+1)$-th point in $p$ is a turn, and over $n_i\ge 0$, otherwise.

Since there are only finitely many checker paths between given points $(0,\tau)$ and $(x,t)$, we get
\begin{align*}
    \sum_{\substack{s: (0,\tau) \rightsquigarrow (x,t)}}|-im\varepsilon|^{\ell(s)}=
    \sum\limits_{\substack{p: (0,\tau) \rightsquigarrow (x,t)}}\sum\limits_{s:p(s)=p} |-im\varepsilon|^{\ell(s)} =\sum_{\substack{p: (0,\tau) \rightsquigarrow (x,t)}}
\frac{(m\varepsilon)^{\mathrm{turns}(p)}}{(1-m\varepsilon)^{\ell(p)}}  < \infty.
\\[-1.2cm]
\end{align*}
\end{proof}

To prove the convergence of the outer series in~\eqref{a+-}, we show that the wave functions $a_{\pm}(x,t;0)$ at times $t=2\varepsilon,3\varepsilon,\dots$ are obtained from the one at the time $\varepsilon$ by iterations of a linear operator $\mathbf{T}$. It is known as the \emph{transfer matrix} and is a classical tool for lattice models \cite{Duminil-Copin-etal-22}. Using a conservation law for the probability (see Figure~\ref{fig: loc_prob}), 
we then show 
that the largest absolute value of eigenvalues of $\mathbf{T}$ is less than $1$, and hence the series converges.

The first step of this plan is the following lemma, proved analogously to Lemma \ref{Dirac_1}.
\begin{lemma}[Reccurence]
  \label{Dirac_1T}
For each lattice point $(x,t)$ with $0<x\le L$ and $t>0$, we have
\begin{align*}
    \begin{pmatrix}
        a_-(x-\varepsilon,t+\varepsilon;0)\\
        a_+(x+\varepsilon,t+\varepsilon;0)
    \end{pmatrix} =
    \underbrace{\frac{1}{1+im\varepsilon}\begin{pmatrix}
1&-im\varepsilon\\
-im\varepsilon&1\\
\end{pmatrix}}_{\text{\large{$U$}}}
    \begin{pmatrix}
    a_-(x,t;0)\\
    a_+(x,t;0)
\end{pmatrix}.
\end{align*}  
\end{lemma}


Introduce the vector space $\mathbb{C}^{D}$ of functions $a^{\circ}(x)= (a_-^{\circ}(x), a_+^{\circ}(x))\colon\varepsilon \mathbb{Z}\to\mathbb{C}^2$ such that $a^{\circ}(x) = (0,0)$ for $x \notin [0,L+\varepsilon]$.
(As usual, we assume $L/{\varepsilon}\in\mathbb{Z}$ so that $D = 2{L}/{\varepsilon}+4$.)
Consider the linear operator $\mathbf{T}$ on this vector space that maps a function $a^{\circ}(x)$ to $b^{\circ}(x):=(b^{\circ}_-(x),b^{\circ}_+(x))$ given by 
\begin{align}
\label{relation1}
\begin{pmatrix}
    b^{\circ}_-(x-\varepsilon)\\
    b^{\circ}_+(x+\varepsilon)\\
\end{pmatrix} :=
\begin{cases}
 U 
\begin{pmatrix}
    a^{\circ}_-(x)\\
    a^{\circ}_+(x)
\end{pmatrix}, & \text{if } x = \varepsilon, \dots, L; \\
\begin{pmatrix}
    0\\
    0
\end{pmatrix}, &\text{otherwise.} 
\end{cases}
\end{align}
Then Lemma~\ref{Dirac_1T}, together with an analog of Example~\ref{ex-bc}, can be rewritten in the form 
\begin{equation}\label{eq-T}
\begin{pmatrix}
    a_-(x,t+\varepsilon;0)\\
    a_+(x,t+\varepsilon;0)
\end{pmatrix}
=\mathbf{T}
\begin{pmatrix}
    a_-(x,t;0)\\
    a_+(x,t;0)
\end{pmatrix}
\qquad \text{for $t= \varepsilon, 2\varepsilon,\dots $.}
\end{equation}

\begin{lemma}[Local probability conservation] (See Figure~\ref{fig: loc_prob}.)
\label{local_prob}
    For each $x = \varepsilon, 2\varepsilon, \dots, L$ we have 
    \begin{equation}
    \label{lprob_eq}
        \begin{aligned}
            |b^{\circ}_+(x+\varepsilon)|^2 + |b^{\circ}_-(x-\varepsilon)|^2 = |a^{\circ}_-(x)|^2 + |a^{\circ}_+(x)|^2
        \end{aligned}
    \end{equation}
\end{lemma}

\begin{proof}
    Indeed, since \eqref{relation1} holds and the matrix $U$ is unitary, it follows that~\eqref{lprob_eq} holds.
\end{proof}

One can think of $|a^\circ_-(x)|^2$ and $|a^\circ_+(x)|^2$ as the `amount of light' entering a lattice point $(x,t)$ from its bottom-right and bottom-left diagonal neighbors respectively, and 
interpret $|b^\circ_{-}(x)|^2$ and $|b^\circ_+(x)|^2$ as analogous quantities for the point $(x,t+\varepsilon)$. See Figure~\ref{fig: loc_prob} to the left. Thus Lemma~\ref{local_prob} means that the `amount of light' entering a lattice point equals the one leaving it. 

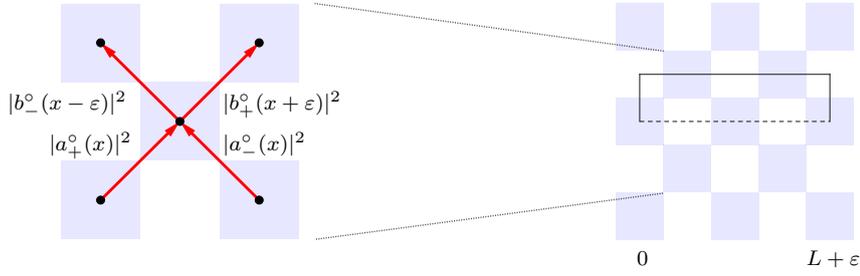
\begin{figure}[h!]
\centering
\begin{overpic}[width=0.66\linewidth]{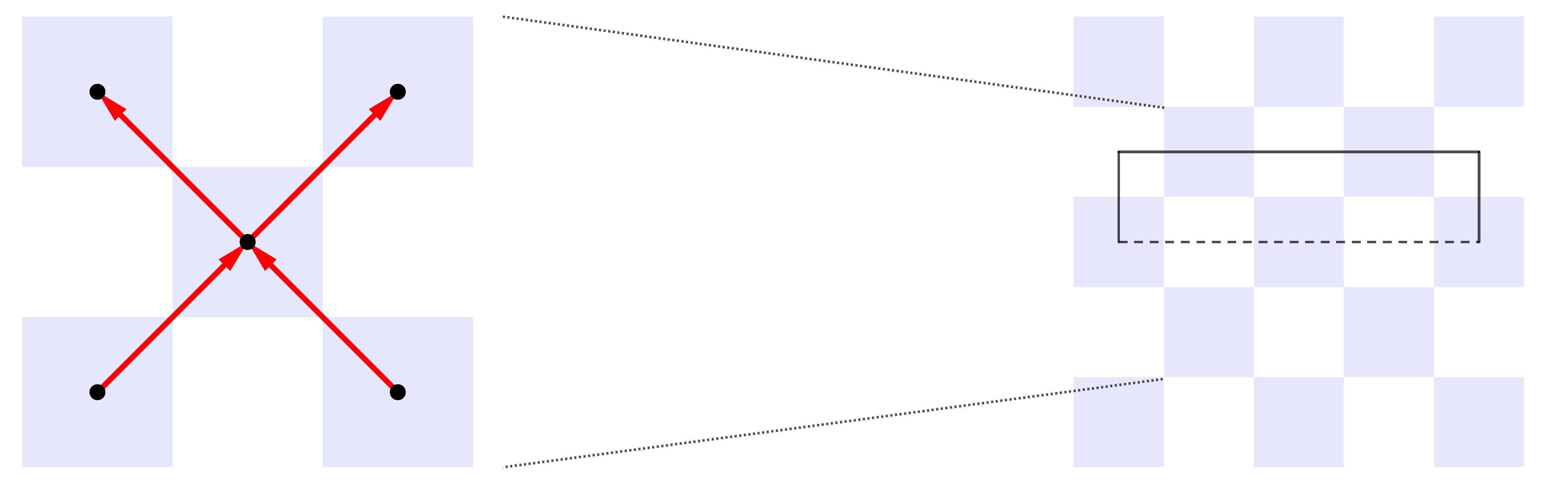}
\footnotesize{
    \put(0,13){$|a_+^{\circ}(x)|^2$}
    \put(-5,18){$|b_-^{\circ}(x-\varepsilon)|^2$}
    \put(21,18){$|b_+^{\circ}(x+\varepsilon)|^2$}
    \put(21,13){$|a_-^{\circ}(x)|^2$}
    \put(71,-1){$0$}
    \put(91.5,-1){$L+\varepsilon$}
}
\end{overpic}
    \caption{(Left) The local probability conservation: the sum of the probabilities that a photon moves to a point from the left and the right equals the sum of the probabilities it moves from the point to the left and the right. See Lemma \ref{local_prob}. (Right) The global probability conservation: the probability of finding a photon on the dashed line equals the probability of finding it on the thick line. We assume that the photon is absorbed at the points $x=0$ and $x=L+\varepsilon$. 
    }
    \label{fig: loc_prob}
\end{figure}

To proceed, consider an eigenvalue $\lambda$ of $\mathbf{T}$ with the maximal absolute value; then $\rho(\mathbf{T}):= |\lambda|$ is called \emph{the spectral radius of $\mathbf{T}$}.

\begin{lemma}[Spectral-radius bound] \label{contraction} The 
operator $\mathbf{T}$ defined by~\eqref{relation1} has spectral radius $\rho (\mathbf{T})<1$.
\end{lemma}

 \begin{proof} Let $\lambda$ be an eigenvalue of $\mathbf{T}$ with the maximal absolute value, so that $\rho(\mathbf{T})=|\lambda|,$ and 
let $$a^{\lambda}=(a^{\lambda}_-(0), a^\lambda_-(\varepsilon), \dots, a^{\lambda}_-(L+\varepsilon), a^{\lambda}_+(0),\dots,a^{\lambda}_+(L+\varepsilon))$$ be a  corresponding eigenvector. 
If $a^{\lambda}_-(0) \ne 0$ then the summation of equations~\eqref{lprob_eq} from Lemma~\ref{local_prob} over all $x$ (see Figure~\ref{fig: loc_prob} to the right)
and the conditions $b^\circ_{+}(0)=b^\circ_{+}(\varepsilon)=b^\circ_{-}(L)=b^\circ_{-}(L+\varepsilon)=0$ given by~\eqref{relation1} lead to 
    \begin{align*}
        \|a^{\lambda}\|^2 = |a_-^{\lambda}(0)|^2 + |a_+^{\lambda}(L+\varepsilon)|^2 + \|\mathbf{T}a^{\lambda}\|^2 > \|\mathbf{T}a^{\lambda}\|^2 = |\lambda|^2 \|a^{\lambda}\|^2,
    \end{align*}
    hence $|\lambda|<1$.
If $a_-^{\lambda}(0) = 0$ then we can rewrite definition~\eqref{relation1} of operator $\mathbf{T}$ 
in the following form:    
        \begin{align}
    \label{eq_15_}
        &\lambda a_+^{\lambda}(x) = \frac{a_+^{\lambda}(x-\varepsilon) - im\varepsilon a_-^{\lambda}(x - \varepsilon)}{1+im\varepsilon}\qquad(x=2\varepsilon,3\varepsilon,\ldots, L+\varepsilon);\\
    \label{eq_16_}
        &a_-^{\lambda}(x)  = \lambda(1+im\varepsilon)a_-^{\lambda}(x-\varepsilon) + im\varepsilon a_+^{\lambda}(x )\qquad(x=\varepsilon,2\varepsilon,\ldots, L).
    \end{align}
    Since $a_+^{\lambda}(\varepsilon)=0$, from \eqref{eq_15_}--\eqref{eq_16_} by induction it follows that $a^{\lambda}$ is identically 
    zero, which contradicts the definition of the eigenvector.
\end{proof}

\begin{lemma}
    \label{summ_conv}
For each 
sequence $a^{\circ}(x,0),a^{\circ}(x,\varepsilon),a^{\circ}(x,2\varepsilon),\dots \in \mathbb{C}^D$ satisfying 
$a^{\circ}(x,t+\varepsilon)=\mathbf{T}a^{\circ}(x,t)$ and each $x \in \varepsilon\mathbb{Z}$, the series 
$
    \sum_{t\in \varepsilon\mathbb{Z},t \ge 0 }
    a^{\circ}_{\pm}(x,t)
$
converges absolutely.
\end{lemma}
\begin{proof}
    We have
\begin{align*}
&
|a^{\circ}_{\pm}(x,t)| \le 
\|a^\circ(x,t)\|=  
\|\mathbf{T}^{t/\varepsilon} a^\circ(x,0)\| \le 
\|\mathbf{T}^{t/\varepsilon}\| \cdot \|a^\circ(x,0)\|.
\end{align*}
Here $\|\cdot\|$ is the Euclidean norm if the argument is a vector in $\mathbb{C}^D$, and the standard operator norm if the argument is an 
operator in $\mathbb{C}^D$.
By Gelfand's formula for the spectral radius and Lemma~\ref{contraction}, we get $
\|\mathbf{T}^k\|^{1/k}\to \rho(\mathbf{T})<1$ as $k\to\infty$.
Thus the series converges absolutely by the Cauchy–Hadamard theorem. 
\end{proof}

\begin{proof}[Proof of Theorem \ref{correctness}] 
The absolute convergence of inner series in~\eqref{a+-} 
is proved in Lemma \ref{inner_conv}. The convergence of the outer series follows from
%
\begin{align*}
    \sum\limits_{\substack{\tau\in \varepsilon \mathbb{Z},\\ \tau<t}} |a_{\pm}(x,t;\tau)|=
    \sum\limits_{\substack{\tau\in \varepsilon \mathbb{Z},\\ \tau<t}} |a_{\pm}(x,t-\tau;0)| = \sum\limits_{\substack{\Delta\in \varepsilon \mathbb{Z},\\ \Delta>0}} |a_{\pm}(x,\Delta;0)| < \infty. 
\end{align*}
Here the first equality holds because $a_{\pm}(x,t;\tau) = e^{i\omega \tau} a_{\pm}(x,t-\tau;0)$ and $|e^{i\omega \tau}|=1$. The second one is a change of the summation index. The last estimate follows from~\eqref{eq-T} and Lemma~\ref{summ_conv}.
%
\end{proof}

Let us give a few final remarks. 
The conservation law (Lemma~\ref{local_prob}) explains the form of equation~\eqref{eq-def-model}, where each scattering contracts and rotates the arrow $a(s)$ through the right angle. This should be the case for transparent materials such as glass, otherwise, the probability conservation 
would not hold, meaning partial absorption (or even emission) of light 
\cite[Figure~69]{Feynman}. This also shows that we must have allowed repeating points in light paths, otherwise, the probability conservation would be again violated. If one tries to fix this by a normalization factor $1/\sqrt{1+m^2\varepsilon^2}$ as in~\eqref{a+-p}, then one gets an incorrect dependence of the reflection probability on the glass thickness in the final formula.

Yet another open problem is finding the eigenvalues and eigenfunctions of the transfer matrix~$\mathbf{T}$.

\anonymize{
\subsection*{Acknowledgements.} 
Sections \ref{experiment}--\ref{proofs} were prepared within the framework of the Academic Fund Program at HSE University in 2023--2024 (grant № 23-00-002).
Sections \ref{Problems}--\ref{appendix_correct} were prepared under support by the Theoretical Physics and Mathematics
Advancements Foundation “Basis” grant 21-7-2-19-1. 
This work was originally presented at the Summer Conference of Tournament of Towns as a collection of problems for pupils \cite{LKTG}. Two participants (less than 16 years old at that time) came up with new results and published them later as appendices in survey~\cite{SU-22}. The work was also presented at 
numerous lectures for pupils 
in 2020-2022. 
We are grateful to all participants, and also to Grigory Chelnokov and Fedor Kuyanov for useful discussions. 
}{}

\anonymizelong

\noindent
\textsc{Fedor Ozhegov\\
HSE University}\\
\texttt{ozhegov.fedor\,@\,gmail$\cdot $com} 

\vspace{0.3cm}
\noindent
\textsc{Mikhail Skopenkov\\
King Abdullah University of Science and Technology}
\\
\texttt{mikhail.skopenkov\,@\,gmail$\cdot $com} \quad \url{https://users.mccme.ru/mskopenkov/}

\vspace{0.3cm}
\noindent
\textsc{Alexey Ustinov\\
HSE University and\\
Khabarovsk Division of the Institute for Applied Mathematics,\\
Far-Eastern Branch,
Russian Academy of Sciences, Russia} 
\\
\texttt{Ustinov.Alexey\,@\,gmail$\cdot $com} \quad
\url{https://www.hse.ru/en/org/persons/530309935}

\endanonymizelong

\end{document}